\documentclass{amsart}

\usepackage{amsfonts}
\usepackage{graphicx}

\newtheorem{theorem}{Theorem}[section]
\newtheorem{corollary}{Corollary}[section]

\newtheorem{lemma}{Lemma}[section]
\newtheorem{Proposition}{Proposition}[section]

\theoremstyle{remark}
\newtheorem{definition}{Definition}[section]
\newtheorem{example}{Example}[section]
\newtheorem{remark}{Remark}[section]

\def\eqref#1{$(\ref{#1})$}
\def\H{\mathbb{H}}
\def\R{\mathbb{R}}
\def\C{\mathbb{C}}
\def\N{\mathbb{N}}
\def\Z{\mathbb{Z}}
\def\F{\Phi}
\def\Wr{W}
\def\wr{\omega}
\def\Mat{M}
\def\mat{m}
\def\cdet{\textup{cdet}}
\def\span{\textup{span}}
\def\QF{\hat\F}
\def\QK{\hat K}
\def\Qc{\hat c}

\def\qi{\textbf{i}}
\def\qj{\textbf{j}}
\def\qk{\textbf{k}}
\def\vi{v_p}
\def\vr{v_c}

\def\hatalpha{\hat\alpha}
\def\hatbeta{\hat\beta}
\def\hatlambda{\hat\lambda}
\def\tildealpha{\tilde\alpha}
\def\tildebeta{\tilde\beta}

\title{Quaternion-Valued Breather Soliton, Rational, and Periodic KdV Solutions}
\author{John Cobb,
Alex Kasman,
Albert Serna, and 
Monique Sparkman}

\address{Department of Mathematics, College of Charleston, Charleston
  SC 29401}

\begin{document}
\voffset-.5in%
\advance\textwidth0in\advance\textheight.5in%
\advance\leftmargin-1in\advance\rightmargin-1in%
\lineskip1.2\lineskip\baselineskip1.2\baselineskip%
\maketitle
\markright{Quaternion-Valued Soliton, Rational and Periodic KdV Solutions}

\begin{abstract}
Quaternion-valued solutions to the non-commutative KdV equation are
produced using determinants.  The solutions produced in this way
are (breather) soliton solutions, rational solutions, spatially
periodic solutions and hybrids of these three basic types.  A complete
characterization of the parameters that lead to non-singular 1-soliton
and periodic solutions is given.  Surprisingly, it is shown that such solutions are never singular when the solution is essentially non-commutative.  When a 1-soliton solution is
combined with another solution through an iterated Darboux
transformation, the result behaves asymptotically like a combination
of different solutions.  This ``non-linear superposition principle''
is used to find a formula for the phase shift in the general 2-soliton
interaction.  A concluding section compares these results with other
research on non-commutative soliton equations and lists some open
questions.
\end{abstract}

\section{Introduction}
\subsection{The KdV Equation}

The 
Korteweg-deVries  (KdV) Equation 
\begin{equation}
u_t=\frac{3}{2}uu_x+\frac{1}{4}u_{xxx}\label{eqn:stdKdV}
\end{equation} was originally derived in order to better understand the
solitary waves observed in 1834 by John Scott Russell on Union Canal in
Scotland \cite{KdV,KdV-hist,GOST}.  It is unusual among
nonlinear partial differential equations in that it is completely
integrable and so it is possible to write many of its solutions in
closed form.  Moreover, among those solutions are the multi-soliton
solutions that behave asymptotically like localized disturbances
traveling at constant speeds which exhibit a phase shift upon
interaction \cite{ZK,BKY}.  Even among other completely integrable differential
equations with soliton solutions, the KdV Equation holds a special
place because it was historically the first one recognized as having
these properties.

In the case that $u(x,t)$ takes values in some non-commutative
algebra, a natural generalization of the KdV Equation is the
symmetrized form:
\begin{equation}
u_t=\frac{3}{4}uu_x+\frac{3}{4}u_xu+\frac{1}{4}u_{xxx}.\label{eqn:KdV}
\end{equation}

 The purpose of this paper is to carefully study certain
 quaternion-valued solutions to \eqref{eqn:KdV}.  Although solutions
 to integrable equations such as KdV have been previously explored both in more general
 non-commutative settings \cite{EGR,NC1,NC2,NC3} and in the 
 quaternionic case \cite{QKdV}, the specific breather soliton,
 rational and periodic solutions investigated below, their
 construction in terms of quaternionic determinants, and their
 nonlinear superpositions have not previously been described.

\subsection{Quaternions}

The quaternions were first studied by William Rowan
Hamilton as a number system which generalized the complex numbers
\cite{Hamilton}.  Although their non-commutativity was a novelty in
1844, since the quaternions can be embedded into a matrix group, they
may not seem particularly interesting to a modern mathematical
physicist.  For many years they were seen as being ``old-fashioned'',
merely a historical stepping stone on the way to more general
non-commutative algebras.  However, recently they have
received an increasing amount of attention in relation to differential
equations and dynamical systems \cite{QDiff1,QDiff2,QDiff3,QDiff4}, for their uses in mathematical
physics and engineering \cite{QMP1,QMP2,QMP3,QMP4}, and even for their
unique algebraic structure \cite{Chen,QMatInv,Kyrchei}.  This
resurgence of interest in the quaternions shows that some
important properties of quaternionic solutions are not immediately
evident when they are viewed in the more general context of matrix
algebras and justifies the current investigation into the
quaternion-valued solutions of the KdV equation.

This section will briefly review some key properties of the quaternions and set up the terminology and notation to be used in the remainder of the paper.  For additional information, readers should consult References \cite{quaternion1,quaternion2}.

\subsubsection{Notation and Arithmetic}
The quaternions are the 4-dimensional real vector space
$$
\H=\{q_0+q_1\qi+q_2\qj+q_3\qk\ :\ q_i\in\R\}
$$
with multiplication satisfying the usual distributive and associative laws along with the identities
$$
\qi^2=\qj^2=\qk^2=\qi\qj\qk=-1
$$
which William Rowan Hamilton famously carved into Brougham Bridge in 1843.  However, the multiplication is not commutative because 
$$
\qi\qj=-\qj\qi,\ \qj\qk=-\qk\qj,\ \hbox{and}\ \qi\qk=-\qk\qi.
$$

If a letter is used to index a quaternion, then the subscripts $0$,
$1$, $2$, and $3$  on the same letter will denote the real numbers
which are its coefficients relative to the basis $\{1,\qi,\qj,\qk\}$
of $\H$. 
The length of a quaternion $q\in\H$ is
defined to be $|q|=\sqrt{q_0^2+q_1^2+q_2^2+q_3^2}$, its quaternionic
conjugate is $q^*=q_0-q_1\qi-q_2\qj-q_3\qk$, and if $q\not=0$ then it has a unique  multiplicative inverse
$
q^{-1}=\frac{1}{|q|}q^*
$.

It is often convenient to separate a quaternion
$q$ into its real part $q_0$ and vector part
$\vec q=q_1\qi+q_2\qj+q_3\qk$ where the latter is thought of as being an
element of $\R^3$.  Then, conjugating one element of $\H$
by another has an interpretation as a \textit{rotation} in
3-dimensional space in the following sense: If
$q$ and $g\not=0$ are quaternions then $q$ and
$
r=gqg^{-1}
$
are related by the facts that $q_0=r_0$ and that 
 $\vec r\in\R^3$ is a vector obtained by rotating
$\vec q$ through an angle depending only on the choice of $g$.
This sets up a well-known correspondence between unit quaternions and
rotations by which every rotation corresponds to either of two 
quaternions of length $1$.  The details of this correspondence will not be needed in
this paper.  However, a certain consequence of its existence will be:
\begin{Proposition}\label{prop:conj}
Two quaternions $q$ and $r$ satisfy $r=gqg^{-1}$ for
some quaternion $g$ if and only if $q_0=r_0$ and $|\vec q|=|\vec r|$.
\end{Proposition}

Exponential functions involving quaternions will be needed later in
this paper.  It is therefore useful to note that by writing $e^q$ as a
power series one can easily show that
\begin{equation}
e^q=e^{q_0}\left(\cos(|\vec q\,|)+\frac{\sin(|\vec q\,|)}{|\vec q\,|}\vec q\right)\label{eqn:exp}
\end{equation}
 for any  quaternion $q=q_0+\vec q$ with non-zero vector part.  Moreover, 
if $q$ and $r$ are quaternions that commute (i.e. $[q,r]=0$) 
then
$
e^{q}e^{r}=e^{q+r}
$.

\subsubsection{Quaternion-Valued Functions}

Throughout the remainder of this paper, $x$ and $t$ will
be real-valued variables and functions of these variables will take
values in $\H$.
Together, such quaternion-valued functions
will be taken to form a right module over the quaternions.  (Hence, any reference
to linear combinations of such functions will be considered to be a
sum of functions with quaternionic coefficients on the right.)

When
 a quaternion-valued function $f(x,t)$ is to be represented graphically, it will be
illustrated by graphing each component function $f_i(x,t)$ ($0\leq i
\leq 3$) separately on the same
set of axes for some fixed value of $t$. 
  So, for instance, the function
 $f(x,t)=(x+t)^{-2}+\sin(x)\qi + \cos(t)\qj+(x^2+1)^{-1}\qk$ is defined
 for all $(x,t)\in\{(x,t)\in R^2\ :\ x\not=-t\}$ and its graphical
 representation at a fixed value of $t$ would look like the
 superposition of four graphs, one having a pole at $x=-t$, a trig
 function, a horizontal line, and a function with one ``peak'' at $x=0$.

\subsubsection{Determinants of Quaternionic Matrices}\label{sec:Chen}

Interestingly, although there is no useful generalization of the
determinant to arbitrary non-commutative settings\footnote{The
  quasi-determinant \cite{EGR} is useful in non-commutative settings.
  However, it is not a \textit{generalization} of the determinant in
  that a quasi-determinant of a matrix which happens to have commuting
  entries is not equal to the determinant of that matrix.}, there
\textit{are} definitions for a determinant of a square matrix of
quaternions that generalizes the usual determinant and have
corresponding Cramer-like theorems \cite{Chen,Kyrchei}. 
By setting up notation and summarizing some prior results, this
section lays the foundation for the construction of quaternion-valued
KdV solutions using
these determinants in Theorem~\ref{thm:kdv}.

\begin{definition}\label{def:cycles}
Let $S_n$ denote the group
of permutations on $n$ elements. 
 A permutation $\sigma\in S_n$ is a
cycle 
$$
\sigma=(c_1c_2\cdots c_k)
$$
if $\sigma(j)=j$ for $j\not\in\{c_1,\ldots,c_k\}$ and
$\sigma(c_i)=c_{i+1}$ for $i<k$ and $\sigma(c_k)=c_1$.  
If a permutation $\sigma$ in the group $S_n$ of permutations on the
set $\{1,\ldots,n\}$ is a cycle, then it 
can be written in the normalized form $\sigma=(c_1c_2\cdots c_k)$ 
where $c_1>c_j$ for $j>1$.  Any permutation
$\sigma\in S_n$ has a unique factorization into normalized cycles
$$
\sigma=\sigma_1\cdots\sigma_r
$$
where for each $j$ one has $\sigma_j=(c^j_1c^j_2\cdots)$ and
$c^j_1>c^{j+1}_1$ (i.e. the sequence of first terms in the cycles
is decreasing) and where each element of $\{1,\ldots,n\}$ appears
exactly once (which requires including cycles of length $1$ for fixed
points of $\sigma$).
\end{definition}
\begin{definition}\label{def:MatrixStuff}
Let $\Mat=[\mat_{ij}]=[\vec \mat_i]$ be an $n\times n$
matrix with entries $\mat_{ij}$ from some non-commutative ring and column vectors $\vec \mat_i$. 
We denote by $\Mat_{\langle i\rangle}$ the matrix
obtained by exchanging the $i^{th}$ and $n^{th}$ columns of $\Mat$:
$$\Mat_{\langle i\rangle}=\left[\vec \mat_1\ \vec \mat_2\ \cdots
 \vec \mat_{i-1}\ \vec
 \mat_n\ \vec\mat_{i+1}\ \cdots\ \vec\mat_{n-1}\ \vec \mat_{i}\right]
$$
and let $\Mat_{\langle i,j\rangle}$ denote the matrix obtained by
replacing the $i^{th}$ column of $\Mat$ by its $n^{th}$ column and
replacing the $n^{th}$ column by the $n$-vector $\vec e_j$ whose only
non-zero entry is a $1$ in the $j^{th}$ position:
$$\Mat_{\langle i,j\rangle}=\left[\vec \mat_1\ \vec \mat_2\ \cdots
 \vec \mat_{i-1}\ \vec
 \mat_n\ \vec\mat_{i+1}\ \cdots\ \vec\mat_{n-1}\ \vec e_{j}\right]
$$
For a cycle $\sigma=(c_{1}\cdots
c_{k})\in S_n$  the symbol $M_{\sigma}$ denotes the ordered product
$$
M_{\sigma}=m_{c_1c_2}m_{c_2c_3}\cdots m_{c_{k-1}c_{k}} m_{c_kc_1}.
$$
\end{definition}

\begin{definition}\label{def:chen}
For an $n\times n$ matrix $\Mat=[\mat_{ij}]$ whose elements are from
some non-commutative ring, define the
  Chen Determinant $\cdet(\Mat)$ to be
$$
\cdet(\Mat)=\sum_{\sigma\in S_n}(-1)^{n-r}
M_{\sigma_1}M_{\sigma_2}\cdots M_{\sigma_r}
$$
where for each permutation
$\sigma=\sigma_1\cdots\sigma_r$ is the decomposition into normalized
cycles in Definition~\ref{def:cycles} and $M_{\sigma_j}$ is defined in Definition~\ref{def:MatrixStuff}.
\end{definition}

Note that if the elements of the matrix mutually commute, then
$\cdet(\Mat)=\det(\Mat)$ is the ordinary determinant of the matrix,
but if they do not then this definition specifies a unique ordering of the
factors.  If the elements $\mat_{ij}\in\H$ are quaternions, then it is
possible to solve the vector equation $\Mat v = w$ or to write the
inverse matrix $\Mat^{-1}$ in terms of these Chen determinants
\cite{Chen,Kyrchei}.  
This construction involves not only the matrix $M$ but also its
conjugate transpose
$\Mat^{\dagger}=[\mat_{ji}^*]$.

\begin{Proposition}\label{prop:inv}
An $n\times n$ matrix $\Mat$ of quaternions is invertible if and only if
the real number $\cdet(\Mat^{\dagger}\Mat)$ is non-zero.  If it is, then the $(i,j)$
entry of the matrix $\Mat^{-1}$ is
$$
\Mat_{ij}^{-1}=\frac{1}{\cdet(\Mat^{\dagger}\Mat)}\cdet\!\left(\Mat_{\langle
    i\rangle}^\dagger\Mat_{\langle i,j\rangle}\right).
$$
\end{Proposition}

\begin{remark}\label{rem:notnew}
Definition~\ref{def:chen} and Proposition~\ref{prop:inv} 
can be found in
References~\cite{Chen,Kyrchei}.  However, they have been
rewritten in the notation set up by Definitions~\ref{def:cycles} and
\ref{def:MatrixStuff} into a form that is more convenient for their
use in this paper.
\end{remark}

\section{Construction of Quaternion-Valued Solutions}

\subsection{KdV-Darboux Kernels}

\begin{definition}\label{def:KdVDarbouxKernel}
Let
$\F=\{\phi_1(x,t),\ldots,\phi_n(x,t)\}$ be a set of functions
$\phi_i:\R^2\to\H$ depending on the real variables $x$ and $t$ and
taking values in the set $\H$ of quaternions. We will call $\F$ a
\textit{KdV-Darboux Kernel} if it has the following properties:
\begin{itemize}
\item\textbf{Dispersion:} For each $1\leq i\leq n$, $\phi_i$ satisfies
  the linear equation
\begin{equation}
\frac{\partial^3\phi_i}{\partial x^3}=\frac{\partial\phi_i}{\partial t}.\label{eqn:dispersion}\end{equation}

\item\textbf{Closure:} For each $1\leq i\leq n$, the second derivative
  $(\phi_i)_{xx}$ is in $\span(\F)$, the right $\H$-module generated by $\F$:
\begin{equation}\frac{\partial^2\phi_i}{\partial x^2}\in\span(\F)=\left\{\sum_{j=1}^n\phi_j\alpha_j\ :\ \alpha_j\in\H\right\}.\label{eqn:closure}\end{equation}
\item\textbf{Independence:}
The $n\times n$ Wronskian matrix
$$
\Wr=[\wr_{ij}]\hbox{ with }
 \wr_{ij}=\frac{\partial^{i-1}\phi_j}{\partial x^{i-1}}$$ satisfies
$\cdet(\Wr^{\dagger}\Wr)\not\equiv0$ and hence is an invertible
matrix by Proposition~\ref{prop:inv}.
\end{itemize}
\end{definition}

  For the purposes of this paper, a differential operator of order $n$ is a polynomial $Q$ in the symbol $\partial$ of the form
  $$
  Q=\sum_{i=0}^n c_i(x,t)\partial^i
  $$ where $c_i(x,t)$ are meromorphic, infinitely differentiable
  functions from $\R^2$ to $\H$ with $c_n(x,t)\not=0$.  We say that
  this operator 
  is monic if $c_n\equiv1$.  These operators act
  on infinitely differentiable functions by the formula
  $$
 Q[f(x,t)]=\left(\sum_{i=0}^n c_i(x,t)\partial^i\right)[f]=\sum_{i=0}^n c_i(x,t)\frac{\partial^i f}{\partial x^i}.
  $$
  The product $Q_1\circ Q_2$ of two differential operators is defined to coincide with their composition as operators: $Q_1\circ Q_2[f(x,t)]=Q_1\left[ Q_2[f(x,t)]\right]$.

\subsection{KdV Solution Associated to a KdV-Darboux Kernel}

\begin{theorem}\label{thm:kdv} 
Let $\F=\{\phi_1,\ldots,\phi_n\}$ be a KdV-Darboux Kernel with
Wronskian matrix $\Wr$.  Then the quaternion-valued
function
\begin{equation}
u_{\F}(x,t)=\left[\frac{2}{\cdet(\Wr^{\dagger}\Wr)}\sum_{i=1}^n\frac{\partial^n\phi_i}{\partial
    x^n}\,\cdet\!\left(\Wr_{\langle i\rangle}^{\dagger}\Wr_{\langle i,n\rangle}\right)\right]_x
\label{eqn:uF}
\end{equation}
is a solution to the non-commutative KdV Equation \eqref{eqn:KdV}.
In the special case that $\F=\{\phi\}$ contains only one element, the
formula simplifies to
\begin{equation}
u_{\F}(x,t)=2\phi_{xx}\phi^{-1}-2\phi_x\phi^{-1}\phi_x\phi^{-1}.\label{eqn:uFone}
\end{equation}
\end{theorem}
\begin{proof} 
As a consequence of the independence property
of Definition~\ref{def:KdVDarbouxKernel}, it follows from Theorem 3.6 in
Reference~\cite{KasOp} that
\begin{equation}
 K=\partial^n-\sum_{i=1}^n\sum_{j=1}^n\frac{\partial^n\phi_i}{\partial
   x^n}\Wr
^{-1}_{ij}\partial^{j-1}.\label{eqn:preKF}
\end{equation}
is the unique monic operator of order $n$ such that $\ker(K)=\span(\F)$.
Using Proposition~\ref{prop:inv}, equation \eqref{eqn:preKF} can be
rewritten in terms of Chen determinants as
\begin{equation}
 K=\partial^n-\frac{1}{\cdet(\Wr^{\dagger}\Wr)}\sum_{i=1}^n\sum_{j=1}^n\frac{\partial^n\phi_i}{\partial x^n}\cdet\!\left(\Wr_{\langle
    i\rangle}^\dagger\Wr_{\langle i,j\rangle}\right)
.\label{eqn:KF}
\end{equation}

The closure property of the KdV-Darboux kernel implies that each
element of $\F$ is in the kernel of the operator $K\circ \partial^2$.
Then, Theorem 5.1 in Reference \cite{KasOp} implies the existence of a
differential operator $L$ satisfying the intertwining relationship
\begin{equation}
K\circ \partial^2=L\circ K.\label{eqn:intertwining}
\end{equation}
Setting up notation for the coefficients of the operators $K$ and $L$,
let us write
$$
K=\partial^n+\sum_{i=0}^{n-1}c_i(x,t)\partial^i\qquad
\hbox{and}
\qquad
L=\partial^2+v(x,t)\partial+u_{\F}(x,t).
$$
Equating coefficients on each side of \eqref{eqn:intertwining} one
finds that $v(x,t)=0$ and $u_{\F}(x,t)=(-2c_{n-1})_x$.  (N.B.\
That the potential in the Schr\"odinger operator $L$ is $-2$ times the
$x$-derivative of the coefficient of $\partial^{n-1}$ in $K$ is a
useful observation which will be
referred to in several of the other proofs in this paper.)
The formula for $u_{\F}(x,t)$ in the claim can then be recovered by
isolating the coefficient $c_{n-1}$ from \eqref{eqn:KF}.

In the case where $\F$ contains only one element, it is clear that
$K=\partial-\phi_x\phi^{-1}$ since this is a monic differential
operator of order $1$ having $\phi$ in its kernel.  But, by the
argument above, this means that $u_{\F}=(2\phi_x\phi^{-1})_x$, which
expands to
the claimed formula.

All that remains is to demonstrate that $u_{\F}$ satisfies the KdV
equation, a fact that follows from the dispersion property of the
KdV-Darboux kernel using a standard technique in soliton theory which
is only briefly outlined below.

Differentiating 
$K(\phi_i)=0$ with respect to $t$, using the
dispersion relation to rewrite $t$ derivatives as $x$ derivatives and
again applying Theorem 5.1 from Reference \cite{KasOp}, one concludes that
$\dot K+K\circ \partial^3=M\circ K$ for some differential operator
$M$. Equating coefficients again one determines that
$M=\partial^3+\frac{3}{2}u\partial+\frac{3}{4}u_x$.  Furthermore,
differentiating \eqref{eqn:intertwining} with respect to  $t$ results
in the Lax equation $\dot L=M\circ L -L\circ M$.  Finally, expanding
out these products of differential operators that Lax equation is seen
to be equivalent to the non-commutative KdV equation \eqref{eqn:KdV}.
\end{proof}

\begin{remark}\label{rem:EGR}
This method of producing solutions to
\eqref{eqn:KdV} and the arguments in the proof are not very different
from those used in the seminal paper by Etingof, Gelfand and Retakh
\cite{EGR} where non-commutative solutions were produced using
\textit{quasi-determinants}.  However, the formula in
Theorem~\ref{thm:kdv} works for all KdV-Darboux kernels, even those
for which the Wronskian matrix contains many zero entries which 
impose obstacles to computing the quasi-determinant.  In addition, we wanted to take
advantage of the extra algebraic structure of the quaternions that
allows for solution of linear systems using the Chen determinant.
\end{remark}


\begin{example}\label{examp:ratl} Let $\F$ be the KdV-Darboux kernel
$$
\F=\left\{x^3+x^2\qi+6t+\qk,6x+2\qi\right\}$$
 and let $\Wr$ be its $2\times 2$ Wronskian
matrix.  There are two terms in the sum in formula \eqref{eqn:uF}, but
the second term will be zero since it has the second derivative of
$6x+2\qi$ as a factor.
So, we multiply
$$
\cdet(\Wr_{\langle 1\rangle}^{\dagger}\Wr_{\langle 1,2\rangle})=
72 x^4-216  x t
+\left(-72 t-48 x^3-8 x
\right)\qi
+12
\qj
+36x
\qk
$$
on the left by $(x^3+x^2\qi+6t+\qk)_{xx}=6x+2\qi$, and multiply that by
the real-valued function
$$
\frac{2}{\cdet(\Wr^{\dagger}\Wr)}=\frac{1}{648 t^2-432  x^3 t+144 t
  x+72 x^6+24 x^4+8 x^2+18}
$$
to get
$$
\frac{-2592  x^2 t+288 t+864 x^5+192 x^3+32 x
+\left(-1728  xt-288  x^4-96 x^2\right)\qi+\left(432 x^2+48\right)\qk
}{324 t^2-216  x^3t+72 xt+36 x^6+12 x^4+4 x^2+9}
$$
The solution $u_{\F}(x,t)$ of 
 \eqref{eqn:KdV} is the $x$-derivative of the expression above.
\end{example}

\subsection{Lemmas Relating Different KdV-Darboux Kernels}

The map associating a KdV-Darboux kernel $\F$ to the corresponding
solution $u_{\F}$ actually depends only on $\span(\F)$:
\begin{lemma}\label{lem:samespan}If $\F$ and $\hat\F$ are KdV-Darboux
  kernels which span the same right $\H$ module, then they produce the
  same KdV solution.
\medskip

\noindent In particular, if $\F=\{\phi_1,\ldots,\phi_n\}$ and
$\hat\F=\{\phi_1q_1,\ldots,\phi_nq_n\}$, for some $q_i\in\H$ such that $q_i\not=0$, then
$
u_{\F}(x,t)=u_{\hat\F}(x,t)$.
\end{lemma}
\begin{proof}
Suppose $\F$ and $\hat\F$ are KdV-Darboux kernels such that
 $\span(\F)=\span(\hat\F)$.  By the Independence property
(cf.\ Definition~\ref{def:KdVDarbouxKernel}), we know that this is a
\textit{free} module and since $\H$ is a division ring, they have the
same dimension, which we will call $n$.  As can be seen in the proof
of Theorem~\ref{thm:kdv},  $u_{\F}=(-2c_{n-1})_x$ where $c_{n-1}$ is
the coefficient of $\partial^{n-1}$ in  the unique monic differential
operator of order $n$ having the elements of $\F$ in its kernel.
However, by assumption, the elements of $\hat\F$ are linear
combinations of those elements of $\F$ with constant coefficients on the right and hence
are also in the kernel of this same operator.  Consequently,
the same operator is associated to $\hat\F$ and the solution
$u_{\hat\F}$ produced from it is also the same.
\end{proof}

However, spanning the same right $\H$ module is not the only way two
KdV-Darboux kernels can correspond to the same solution.  The
following lemma shows that they do not even have to have the same
number of elements:
\begin{lemma}\label{lem:dropone}
If $\F=\{\phi_1,\ldots,\phi_n\}$ is a KdV-Darboux kernel and
$\phi_n=\alpha e^{\lambda x+\lambda^3t}$ for some $\alpha,\lambda\in\H$ 
then
$u_{\F}(x,t)=u_{\QF}(x,t)$ where
$$
\QF=\{Q(\phi_1),\ldots,Q(\phi_{n-1})\}\qquad \hbox{and}\qquad
Q(f)=f_x-\alpha \lambda\alpha^{-1}f.
$$
\end{lemma}
\begin{proof}
Let $\QK=\partial^{n-1}+\sum_{i=0}^{n-2}\Qc_i(x,t)\partial^i$ be the
unique monic differential operator of order $n-1$
 having
the elements of $\QF$ in its kernel.  Then we know from the proof of
Theorem~\ref{thm:kdv} that $u_{\QF}=(-2c_{n-2})_x$.

Let $Q=\partial-\alpha \lambda \alpha^{-1}$ be the monic differential
operator of order $1$ with $\phi_n$ in its kernel.  Define
$K=\QK\circ Q$ and note that
$K=\partial^n+\sum_{i=0}^{n-1}c_n(x,t)\partial^i$ is a monic
differential operator of order $n$.  Now consider
$$
K(\phi_i)=\QK\circ Q(\phi_i)=\QK(Q(\phi_i)).
$$
For $i<n$ it is zero since $\QK$ was constructed so that $Q(\phi_i)$
is in its kernel and for $i=n$ this is zero because $Q(\phi_n)=0$.
Then $K$ must be the unique monic differential operator of order $n$
having the elements of $\F$ in its kernel and $u_{\F}=(-2c_{n-1})_x$.

Expanding the product $\QK\circ Q$ we find that the coefficient of
$\partial^{n-1}$ is $c_{n-1}=\Qc_{n-2}-\alpha^{-1}\lambda\alpha$.
Since the second term is constant and has derivative equal to zero, we
conclude that $u_{\F}=(-2c_{n-1})_x=(-2\Qc_{n-2})_x=u_{\QF}$.
\end{proof}
For example, for any quaternions $\alpha$ and $\lambda$ (with
$\alpha\not=0$) the two-element KdV-Darboux kernel
$\F=\{x,\alpha e^{\lambda x+\lambda^3t}\}$ and the single-element
KdV-Darboux kernel $\hat\F=\{1-\alpha^{-1}\lambda\alpha x\}$ produce
the same solution.

Finally, we note that multiplying every element of the KdV-Darboux
kernel on the left by the same non-zero quaternion has the effect of
\textit{rotating} the corresponding solution:
\begin{lemma}\label{lem:rot}
Let $\F=\{\phi_1,\ldots,\phi_n\}$ be a KdV-Darboux kernel, 
$q\in\H$ a non-zero quaternion, and
$\hat\F=\{q\phi_1,\ldots,q\phi_n\}$. Then the solutions $u_{\F}(x,t)$ and
$u_{\hat\F}(x,t)$ are related by the formula
$$
u_{\hat\F}=qu_{\F}q^{-1}.
$$
\end{lemma}
\begin{proof}
Let $K$ be the monic differential operator of order $n$ having the
elements of $\F$ in its kernel.  Note that $\hat K=qKq^{-1}$ is a monic
differential operator of order $n$ and that $\hat
K(q\phi_i)=qKq^{-1}(q\phi_i)=qK(\phi_i)=0$.  Consequently, $\hat K$ is
the unique monic differential operator of order $n$ having the
elements of $\hat \F$ in its kernel.  Letting $c_{n-1}(x,t)$ and $\hat
c_{n-1}(x,t)$ denote the coefficient of $\partial^{n-1}$ in $K$ and
$\hat K$ respectively we have by the definition of $\hat K$ that $\hat
c_{n-1}=qc_{n-1}q^{-1}$.  Then $$u_{\hat \F}=(-2\hat c_{n-1})_x=(-2q c_{n-1}q^{-1})_x=q(-2c_{n-1})_xq^{-1}=qu_{\F}q^{-1}.$$
\end{proof}

\section{Basic Solution Types}

There are three kinds of non-trivial quaternion-valued solutions to
\eqref{eqn:KdV} that can be produced using a KdV-Darboux kernel with
one element: localized breather solitons, translating periodic
solutions, and rational solutions.  

\subsection{$1$-Soliton and Translating Periodic Solutions}\label{sec:uabl}

Section~\ref{sec:uabl} will consider the solutions
associated to  KdV-Darboux kernels
of the form $\{\phi_{\alpha,\beta,\lambda}\}$ where
\begin{equation}
\phi_{\alpha,\beta,\lambda}(x,t)=\alpha e^{\lambda x+\lambda^3t}+\beta
e^{-\lambda x-\lambda^3t}\label{eqn:phiabl}
\end{equation}
for some choice of $\alpha$, $\beta$ and $\lambda$ in $\H$.  For
convenience, we will write $u_{\alpha,\beta,\lambda}(x,t)$ for the
corresponding KdV solution
$$
u_{\alpha,\beta,\lambda}(x,t)=u_{\{\phi_{\alpha,\beta,\lambda}\}}.
$$

In fact, it is not necessary to consider \textit{all} possible
combinations of quaternions $\alpha$, $\beta$ and $\lambda$.
First, we will assume that 
$\alpha\beta\lambda \not=0$.  This both guarantees that $\{\phi_{\alpha,\beta,\lambda}\}$ is a
KdV-Darboux kernel (which fails to be the case when $\alpha=\beta=0$)
and eliminates the cases in which
$u_{\alpha,\beta,\lambda}(x,t)\equiv0$ is the trivial solution.

Furthermore, one can greatly restrict the selection of the parameter
$\lambda$ without losing any corresponding KdV solutions.
\begin{lemma}\label{lem:complexlambda}
Let $\alpha$, $\beta$ and $\lambda$ be quaternions such that
$\alpha\beta \lambda\not=0$.  Then there are quaternions $\hat\alpha$ and
$\hat\beta$ and a complex number $\hat \lambda=\hat \lambda_0+\hat \lambda_1\qi$ with
$\hat\lambda_0\geq 0$ and $\hat\lambda_1\geq0$ such that
$$
u_{\alpha,\beta,\lambda}(x,t)=u_{\hat \alpha,\hat\beta,\hat \lambda}(x,t).
$$
\end{lemma}
\begin{proof}
Since $\phi_{\alpha,\beta,\lambda}=\phi_{\beta,\alpha,-\lambda}$ we
may assume, without loss of generality, that $\lambda_0\geq0$.

Let $\lambda=\lambda_0+\vec \lambda$ be the decomposition of $\lambda$ into its real and
vector parts.  It will be shown that the same solution can be
constructed using the complex number $\hat \lambda=\lambda_0+|\vec \lambda|\qi$
which has a non-negative real and imaginary components.

By Proposition~\ref{prop:conj}, because $\lambda$ and $\hat \lambda$ have the same real part and
vector parts of the same length, there is a non-zero quaternion $g$
which satisfies
$$
\hat \lambda=g\lambda g^{-1}.
$$
Define $\hat\alpha=\alpha g^{-1}$ and $\hat \beta=\beta g^{-1}$.
Now, note that 
\begin{eqnarray*}
\phi_{\hat\alpha,\hat\beta,\hat \lambda}
&=&\hat \alpha e^{\hat \lambda x+\hat \lambda^3t}+\hat \beta e^{-\hat
  \lambda x-\hat \lambda^3t}\\
&=& \alpha g^{-1} e^{g \lambda g^{-1} x+g\lambda^3g^{-1}t}+ \beta g^{-1} e^{- g \lambda
  g^{-1}x-g \lambda^3g^{-1}t}\\
&=& \alpha g^{-1}(g e^{ \lambda x+\lambda^3t}g^{-1})+ \beta g^{-1}(g e^{- \lambda x-
  \lambda^3t}g^{-1})\\
&=& \left(\alpha  e^{ \lambda x+\lambda^3t}+ \beta  e^{- \lambda x-
  \lambda^3t}\right)g^{-1}
=\phi_{\alpha,\beta,\lambda}g^{-1}.
\end{eqnarray*}
It then follows from Lemma~\ref{lem:samespan} that they generate
the same solutions.
\end{proof}

Consequently, no non-trivial solutions will be lost by the fact that
we will henceforth limit our attention only
to the case in which $\alpha\beta\lambda \not=0$ and
$\lambda=\lambda_0+\lambda_1\qi$ is a
complex number with $\lambda_0,\lambda_1\geq 0$. 

The real numbers 
$$\vr=\lambda_0^2-3\lambda_1^2\qquad \hbox{and}\qquad
\vi=3\lambda_0^2-\lambda_1^2$$
 will then
be useful in understanding the solution any of the solutions
$u_{\alpha,\beta,\lambda}$ since by Equation~\ref{eqn:exp}
\begin{equation}
e^{\lambda x+\lambda^3t}=e^{\lambda_0(x+\vr t)}\left(\cos(\lambda_1(x+\vi
t))+\sin(\lambda_1(x+\vi t))\qi\right).\label{eqn:velocities}
\end{equation}
As one might guess from \eqref{eqn:velocities}, $\vr$ and $\vi$ will
play the role of two separate \textit{velocities}.  Considering the
graph of $u_{\alpha,\beta,\lambda}$ as a function of $x$ with $t$
playing the role of a time parameter, the \textit{periodic} features coming
from the trigonometric functions will translate to the left with
velocity $\vi$ while the localized soliton has a \textit{center} which
translates with velocity $\vr$.

The next two sections separately handle the cases $\lambda_0=0$ and
$\lambda_0\not=0$ which are qualitatively very different.

\subsubsection{Translating Periodic Solutions}

Consider the case in which $\lambda_0=0$ (so that $\lambda=\lambda_1\qi$ is a purely
imaginary complex number).  Then the corresponding solution to the KdV
equation is a spatially periodic solution that translates at a
constant speed in time.

\begin{theorem}
If $\lambda=\lambda_1\qi$ then the associated KdV solution 
has a graph that is invariant under a horizontal translation in $x$ by
$2\pi/\lambda_1$ units and viewing $t$ as a time parameter this
periodic waveform translates to the right at constant speed $\lambda_1^2$.
\end{theorem}
\begin{proof}
Using Equation~\ref{eqn:velocities}
and the fact that $\vi=3\lambda_0^2-\lambda_1^2=-\lambda_1^2$ one finds that
$$
\phi_{\alpha,\beta,\lambda}(x,t)=
(\alpha+\beta)\cos(\lambda_1(x-\lambda_1^2
t))+(\alpha-\beta)\sin(\lambda_1(x-\lambda_1^2 t))\qi.
$$

Substituting this into \eqref{eqn:uFone} and using trigonometric
identities, one can determine that the corresponding KdV solution has the form
$$
u_{\alpha,\beta,\lambda}(x,t)=R(\cos(2\lambda_1(x-\lambda_1^2t)),\sin(2\lambda_1(x-\lambda_1^2t)))
$$
where $R(\xi,\eta)$ is a certain rational function.

Since $u_{\alpha,\beta,\lambda}$ can be written as a function in
$x-\lambda_1^2t$, we know that its graph as a function of $x$ will
translate to the right with speed $\lambda_1^2$ in the time parameter
$t$.  And, since the translation $x\mapsto x+\pi/\lambda_1$ shifts the
arguments of the trigonometric functions by $2\pi$, it leaves the
graph unchanged.
\end{proof}

\begin{figure}

\centering

\includegraphics[width=3in,height=2in]{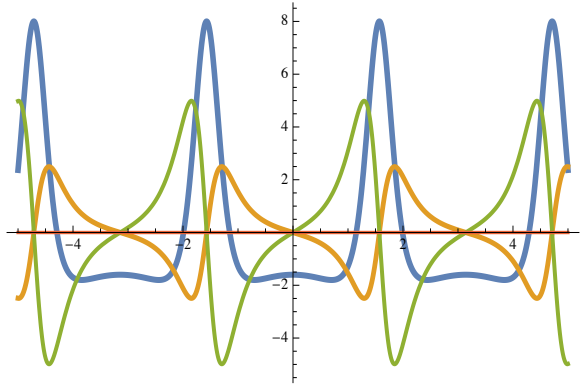}

\caption{The periodic translating solution
  $u_{\alpha,\beta,\lambda}(x,t)$ with
  $\alpha=1+\qk$, $\beta=1$ and $\lambda=\qi$ at time $t=0$.}\label{fig:periodic}

\end{figure}

\begin{example}\label{examp:periodic} If $\alpha=1+\qk$, $\beta=1$ and $\lambda=\qi$ then
$$
\phi_{\alpha,\beta,\lambda}(x,t)=2\cos(x-t)+\sin(x-t)\qj+\cos(x-t)\qk
$$
and $u_{\alpha,\beta,\lambda}(x,t)=w(x-t)$ where
$$
w(\xi)=
-\frac{8 (3 \cos (2 (\xi))+2)}{(2 \cos (2 (\xi))+3)^2}
-\frac{8 \sin (2 (\xi))}{(2 \cos (2 (\xi))+3)^2}\qi
+\frac{16 \sin (2 (\xi))}{(2 \cos (2 (\xi))+3)^2}\qj.
$$
The four components of this solution at time $t=0$ are shown in
Figure~\ref{fig:periodic}.  As expected, an animation shows the
solution translating to the right at constant speed $1$ and a
horizontal spatial
translation by $\pi$ units leaves the graph of $w(\xi)$ invariant.
\end{example}

\subsubsection{Localized Breather Solitons}

\begin{theorem}\label{thm:breathers}
If $\lambda_0\not=0$ then for any fixed $t$ the graph of the solution
$u_{\alpha,\beta,\lambda}$ as a function of $x$ will be localized in a
small neighborhood of 
\begin{equation}
x=c_{\alpha,\beta,\lambda}(t)=\frac{\ln|\alpha^{-1}\beta|}{2\lambda_0}-\vr t\label{eqn:center}
\end{equation}
and hence this disturbance is moving to the left with velocity
$\vr=\lambda_0^2-3\lambda_1^2$.  Moreover, if $\lambda_1\not=0$ then
the solution will also exhibit fluctuations moving to the left at
velocity $\vi=3\lambda_0^2-\lambda_1^2$, giving it the
``throbbing'' appearance of a \textit{breather soliton}.
\end{theorem}
\begin{proof}
The squared amplitude of the solution
$u_{\alpha,\beta,\lambda}$ can be written
 in the form
$$
|u_{\alpha,\beta,\lambda}(x,t)|^2=
\frac{C_1}{\left(|\alpha|^2e^{\Lambda}+|\beta|^2e^{-\Lambda}+C_2
  \cos(\Theta)+C_3\sin(\Theta)\right)^2}
$$
where $\Lambda=2\lambda_0(x+\vr t)$, $\Theta=2 \lambda_1(x+\vi t)$
and $C_1$, $C_2$ and $C_3$ are some constants.

If $\lambda_0>0$ then for sufficiently large values of $|x|$ this amplitude
converges quickly to $0$.  In this sense, we can already see that the
solution is \textit{localized} when $\lambda$ has a non-zero real
part.  Moreover, if we ``average out'' the small variation from the
trigonometric functions by setting them both equal to zero, then this
amplitude function has a unique local maximum located at $x=c_{\alpha,\beta,\lambda}(t)$.

So, in the case $\lambda_0>0$ an animation of the solution
$u_{\alpha,\beta,\lambda}(x,t)$ as a function of $x$ with $t$ playing
the role of time will show a localized disturbance centered at
$x=c_{\alpha,\beta,\lambda}(t)$ and traveling to the left with
velocity $\vr$.  However, if $\lambda_1>0$ as well, then the formula
for the solution will also involve at least one of the functions
$\cos(\theta)$ or $\sin(\theta)$.  Since $\theta$ is a function of
$x+\vi t$, these features will be moving to the left at velocity
$\vi$.  If it were the case that $\vr=\vi$, then the waveform would
simply translate in time.  However, there are no real solutions to
$3\lambda_0^2-\lambda_1^2=\lambda_0^2-3\lambda_1^2$ 
and so whenever $\lambda_0$ and $\lambda_1$ are both non-zero, an
animation of the solution will exhibit the ``breathing'' phenomenon.
\end{proof}

\begin{figure}
\centering

\hbox{\includegraphics[width=2.4in,height=2.4in]{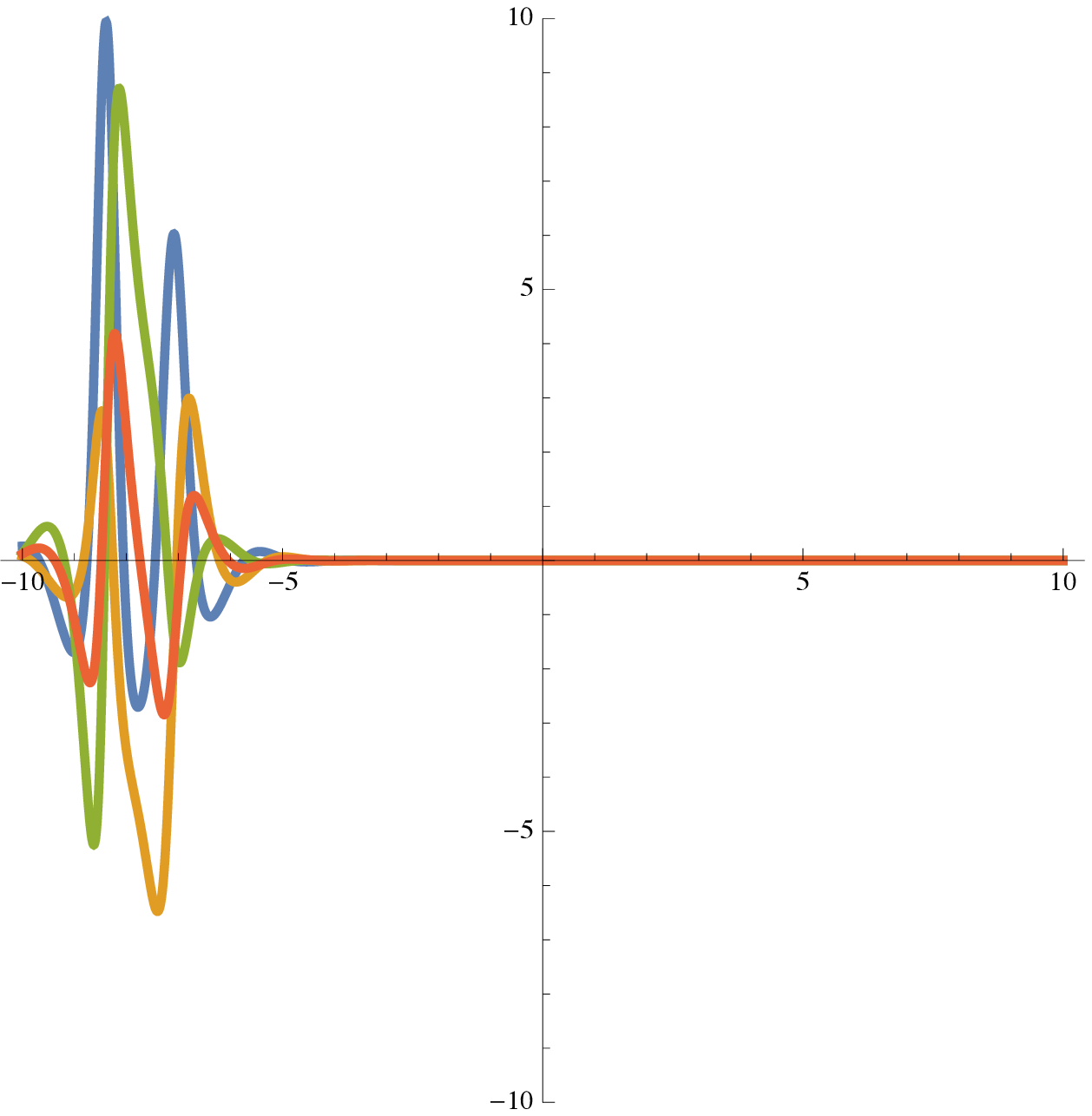}
\qquad
\includegraphics[width=2.4in,height=2.4in]{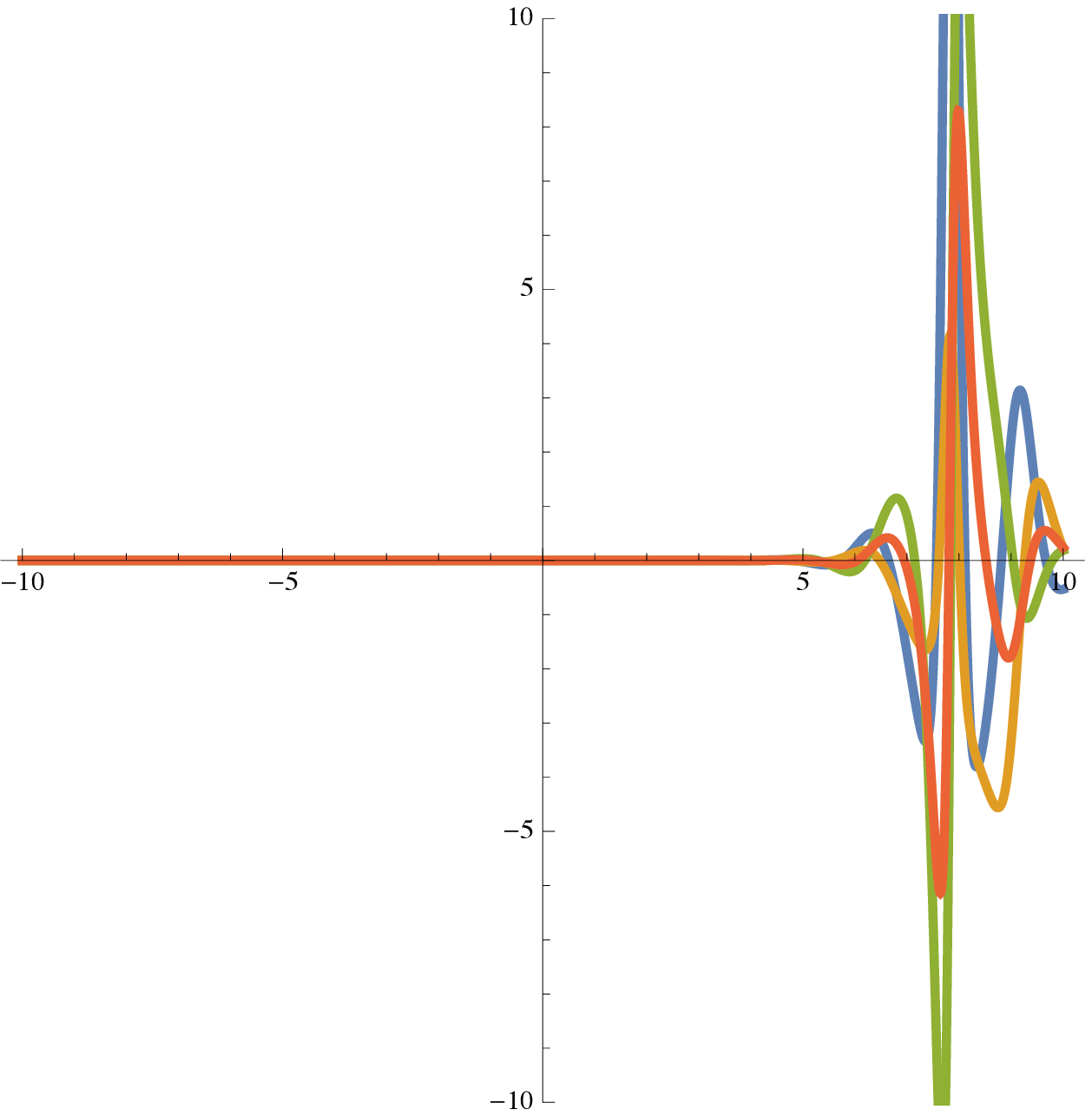}}

\caption{The $1$-soliton solution $u_{\alpha,\beta,\lambda}$ with
  $\alpha=\qj$, $\beta=1+\qk$ and $\lambda=1+\sqrt{3}\qi$ at times
  $t=-1$ (left) and $t=1$ (right).}\label{fig:1sol}
\end{figure}

\begin{example}\label{examp:1sol} Consider the case $\alpha=\qj$, $\beta=1+\qk$ and
$\lambda=1+\sqrt{3}\qi$.  We expect to see a localized disturbance
traveling to the left with velocity $\vr=-8$ (which means it will move
to the right with speed $8$ as $t$ increases).  Moreover, since
$\vi=0$ all of the trigonometric function in the formula for
$u_{\alpha,\beta,\lambda}(x,t)$ have arguments that are independent of
$t$.  At time $t$ we would expect to see the center of the
localized hump being located at $ x=8t+\ln(2)/4 $ (just
slightly to the right of $x=8t$).  And this is what we see in the figures
showing this solution at times $t=-1$ and $t=1$.
\end{example}

\subsubsection{Singularities}

Either the periodic or breather soliton solutions can exhibit
singularities.  The following result completely identifies the
values of the parameters $\alpha$, $\beta$ and $\lambda$ that produce
entirely non-singular solutions.  A surprising corollary is that any
of these solutions exhibiting singularities must be inherently
\textit{commutative} in that it is conjugate to a complex-valued
solution of the standard KdV equation.

\begin{theorem}\label{thm:sing}
The quaternion-valued KdV solution $u_{\alpha,\beta,\lambda}(x,t)$ is
non-singular for all $x$ and $t$ if any of these three conditions involving
the number $q=\alpha^{-1}\beta$ is satisfied:
\begin{itemize}
\item[I.] $q_2^2+q_3^2>0$ (i.e. $q$ is not a complex number)
\item[II.] $\lambda_0=0$ and $|q|\not=1$
\item[] or
\item[III.] $\lambda_1=0$ and $q$ is not a negative
  real number.
\end{itemize}
Moreover, the solution $u_{\alpha,\beta,\lambda}(x,t)$ is undefined
for some $(x,t)\in\R^2$ if none of the conditions is satisfied.
\end{theorem}
\begin{proof}
Throughout Section~\ref{sec:uabl} it has been assumed that
$\alpha\beta\lambda\not=0$.  Hence, we know that $\alpha$ is
invertible.  Define $q=\alpha^{-1}\beta$ and note that by
Lemma~\ref{lem:rot} we know that $u_{1,q,\lambda}=\alpha^{-1}
u_{\alpha,\beta,\lambda}\alpha$.  One of these solutions is singular
if and only if the other is.  Consequently, it is sufficient to
determine when the solution $u_{1,q,\lambda}$ is singular.

Since $\phi_{1,q,\lambda}$ is infinitely differentiable for
all $x$ and $t$ regardless of the values of the parameters, the only
way that $u_{1,q,\lambda}$ given by \eqref{eqn:uFone} could
fail to be defined and differentiable at $(x_0,t_0)$ is if
$|\phi_{1,q,\lambda}|(x_0,t_0)$ is zero.  Using the previous notation
that $\vi=3\lambda_0^2-\lambda_1^2$, $\vr=\lambda_0^2-3\lambda_1^2$ and now defining
$\theta=2\lambda_1(x+\vi t)$ we find
\begin{equation}
|\phi_{1,q,\lambda}(x,t)|^2=(q_0+e^{2\lambda_0(x+\vr
  t)}\cos(\theta))^2+(q_1+e^{2\lambda_0(x+\vr t)}\sin(\theta))^2+q_2^2+q_3^2.
\label{eqn:lensqphi}
\end{equation}
This sum of the squares of four real numbers can only be zero if all
four of them are equal to zero.

If condition I is met then the last two terms in
\eqref{eqn:lensqphi} are non-zero and the solution must be non-singular.

Now suppose condition II is met.
Since $\lambda_0=0$ the exponential terms are equal to $1$ and \eqref{eqn:lensqphi}
reduces to
$(q_0+\cos(\theta))^2+(q_1+\sin(\theta))^2+q_2^2+q_3^2$. However,
this is only zero if $(-q_0,-q_1)$ are the coordinates of the point at
angle $\theta$ radians on the unit circle.  Since $|q|\not=1$ that
cannot be true.

And, if condition III is met then $\theta=0$ and \eqref{eqn:lensqphi}
becomes
$$
|\phi_{1,q,\lambda}(x,t)|^2=(q_0+e^{2\lambda_0(x+\vr
  t)})^2+q_1^2+q_2^2+q_3^2.
$$
If $q$ is not a real number then $q_1^2+q_2^2+q_3^2>0$ and the
expression is non-zero  And if $q$ is a non-negative real number then
the first term is positive for all values of $x$ and $t$.

This shows that the solution is entirely non-singular if any one of
the three conditions is met.  Now, assume that none of the conditions
is satisfied.  So, we know that $q_2^2+q_3^2=0$, that either
(i) $\lambda_0\not=0$ or that (ii) $|q|=1$, and that either (iii) $\lambda_1\not=0$ or
that (iv) $q$ is negative.

Assuming (i) and (iii) ensures that the function in the exponent and
the trigonometric argument $\theta$ are linearly independent linear
functions of the variables $x$ and $t$ and hence can be simultaneously
solved to take any desired value.  The
point $(-q_0,-q_1)$ is on a circle of radius $|q|$ around the origin
and hence can be written as $(|q|\cos(\theta_0),|q|\sin(\theta_0))$
for some value of $\theta$.  Then one can simultaneously find
values of $x$ and $t$ such that $e^{2\lambda_0(x+\vr t)}=|q|$ 
and $\theta=\theta_0$ thereby making the entire
expression equal to zero.

If (i) and (iv) are assumed to be true then 
$q$ is a negative real number and we want to show that
$$(q+e^{2\lambda_0(x+\vr
  t)}\cos(\theta))^2+(e^{2\lambda_0(x+\vr t)}\sin(\theta))^2
$$
is zero for some choice for $x$ and $t$.  If $\theta=0$ then the
second term is already zero and if not then a value of $t$ can be
found to make it zero.  Either way the expression then reduces to
$(q+e^{2\lambda_0(x+ C)})^2$ for some constant $C$ and 
the expression is equal to zero at
$x=(\ln(-q)-C)/(2\lambda_0)$ (which is a real number as a consequence
of (i) and (iv)).

Now suppose that (ii) and (iii) are true.  We can further assume that
$\lambda_0=0$ because otherwise (i) is true and we already handled that
case.  But then the expression reduces to
$(q_0+\cos(\theta))^2+(q_1+\sin(\theta))^2$.  By assumption (ii), we
know that the point $(-q_0,-q_1)$ lies on the unit circle and hence
there is a number $\theta_0$ such that it is equal to
$(\cos(\theta_0),\sin(\theta_0))$.  Since $\lambda_1\not=0$ it is possible
to choose $x$ and $t$ so that $\theta=\theta_0$ and the expression
then becomes zero.

Finally, consider the case in which (ii) and (iv) are both true.  If
$\lambda_0\not=0$ is also true then that would mean (i) and (iv) are
true, and it has already been demonstrated that the solution is
singular in that case.  On the other hand, if $\lambda_0=0$ then
$\lambda_1\not=0$ (because $\lambda\not=0$ is assumed throughout this
section),  but then (ii) and (iii) are true which has also already
been handled.  
\end{proof}

\begin{remark}\label{rem:non-sing-hard}
One might guess from Figure~\ref{fig:1sol} that the breather soliton
solution in
Example~\ref{examp:1sol} is singular because it appears to have a pole
in the bottom figure.  However, 
$\alpha^{-1}\beta=-\qi-\qj$ is not a complex number and hence
according to Theorem~\ref{thm:sing} it is not.  (In fact, redrawing
the graph at time $t=1$ over a larger vertical range confirms that
there is simply a local maximum that is outside of the viewing window
in Figure~\ref{fig:1sol}.)
\end{remark}

\begin{example}\label{singular} The periodic solution shown in
Example~\ref{examp:ratl} is non-singular because $\alpha=1+\qk$,
$\beta=1$ and $\lambda=\qi$ so $q=\alpha^{-1}\beta=1/2-1/2\qk$ which
satisfies criterion II.  On the other hand, choosing $\alpha=1$,
$\beta=1/\sqrt{5}-2/\sqrt{5}\qi$ and $\lambda=\qi$ results in a singular solution $u_{\alpha,\beta,\lambda}(x,t)$ since none of the conditions are satisfied.
This particular solution may seem uninteresting as it is
complex-valued and therefore not inherently non-commutative, but it
will play an important role in Example~\ref{examp:onesidesing} below.
\end{example}

Surprisingly, it turns out that $u_{\alpha,\beta,\lambda}$ has
singularities \textit{only} when the solution is really commutative:
\begin{corollary}\label{cor:sing}
If the solution $u_{\alpha,\beta,\lambda}(x,t)$ is singular then
$u(x,t)=\alpha^{-1}u_{\alpha,\beta,\lambda}\alpha$ is a
complex-valued function, and $u_{\alpha,\beta,\lambda}$ is a solution of the usual KdV equation \eqref{eqn:stdKdV}.
\end{corollary}
\begin{proof}
If $u_{\alpha,\beta,\lambda}(x,t)$ is singular then
$q=\alpha^{-1}\beta\in\C$ must be a complex number (else
Condition I of
Theorem~\ref{thm:sing} is met and the solution would be
non-singular).  Note that $\phi_{1,q,\lambda}=\alpha^{-1}\phi_{\alpha,\beta,\lambda}$.
By Lemma~\ref{lem:rot}
$$
u_{1,q,\lambda}(x,t)=\alpha^{-1}u_{\alpha,\beta,\lambda}(x,t)\alpha=u(x,t)
$$
is another solution to \eqref{eqn:KdV}.
However, since $q$ and $\lambda$ are both complex numbers,
$\phi_{1,q,\lambda}$ and therefore $u_{1,q,\lambda}$ which can be
computed from it using \eqref{eqn:uFone} are complex-valued
functions.  Consequently $u$ and $u_x$ commute.  Commutativity is
preserved by conjugation so $u_{\alpha,\beta,\lambda}$ also commutes
with its derivative.  Then, both of these functions solve
\eqref{eqn:stdKdV}.
\end{proof}

\subsection{Rational Solutions}

\begin{definition}\label{def:Delta}
Let $\psi_0(x,t,z)=e^{xz+tz^3}$ and for $m=0,1,2,\ldots$ define
$\Delta_m(x,t)$ to be 
$$
\Delta_m(x,t)=\frac{\partial^m\psi_0}{\partial z^m}\bigg|_{z=0}.
$$
Note that $\Delta_m(x,t)\in\Z[x,t]$ is a polynomial in $x$ and $t$
with integer coefficients and that it has degree $m$ as a polynomial
in $x$.
\end{definition}
\begin{theorem}\label{thm:rat}
For any $n\in\N$ and $\alpha_j\in\H$,
$\F=\{\phi_0,\ldots,\phi_n\}$ is a KdV-Darboux kernel where
$$
\phi_i(x,t)=\frac{\partial^{2i}}{\partial
  x^{2i}}\left(\Delta_{2n+1}(x,t)+\sum_{j=0}^{n}\Delta_{2j}(x,t)\alpha_j\right).
$$
 Each component function $u_i(x,t)$ of corresponding quaternion-valued solution
 $u_{\F}(x,t)=u_0(x,t)+u_1(x,t)\qi+u_2(x,t)\qj+u_3(x,t)\qk$ 
 to \eqref{eqn:KdV} is a rational function.
\end{theorem}
\begin{proof}
Because
$$
\frac{\partial^3}{\partial x^3}\Delta_m(x,t)
=
\frac{\partial^3}{\partial x^3}\frac{\partial^i}{\partial z^i}\psi_0\bigg|_{z=0}
=
\frac{\partial^i}{\partial z^i}\frac{\partial^3}{\partial x^3}\psi_0\bigg|_{z=0}
=
\frac{\partial^i}{\partial z^i}z^3\psi_0\bigg|_{z=0}
$$
$$
=
\frac{\partial^i}{\partial z^i}\frac{\partial}{\partial t}\psi_0\bigg|_{z=0}
=
\frac{\partial}{\partial t}\frac{\partial^i}{\partial z^i}\psi_0\bigg|_{z=0}
=\frac{\partial}{\partial t}\Delta_m(x,t),$$
by linearity each function $\phi_i$ also satisfies the dispersion
condition of Definition~\ref{def:KdVDarbouxKernel}.

For $i<n$, $(\phi_i)_{xx}=\phi_{i+1}\in\F$.  On the other hand, $\phi_n$ is the
$(2n)^{th}$ derivative of a polynomial of degree $2n+1$, and so it is a
linear function of $x$.  Then $(\phi_n)_{xx}=0\in\span(\F)$.  Thus
$\F$ satisfies the closure property for KdV-Darboux Kernels.

Finally, to demonstrate the invertibility of the Wronskian matrix $W$,
we consider a linear combination
$$
\sum_{i=0}^n\phi_i\alpha_i
$$
of the entries in its first row with quaternionic coefficients on the
right.  Suppose that not all of the coefficients in this linear
combination are zero and let $j$ be the smallest value of $i$ for
which $\alpha_i\not=0$.  Then, since $\phi_i$ is a polynomial in $x$
of degree $2(n-i)+1$ the $x^{2(n-j)+1}$ term which comes from
$\phi_j\alpha_j$ cannot be cancelled by any other terms in the sum.
Then the linear combination is not equal to zero unless all of the
coefficients are zero and the homogenous vector equation $Wv=0$ has no
non-trivial solutions which implies the invertibility of the matrix \cite{QMatInv}.

Since $\F$ is a KdV-Darboux kernel we know that $u_{\F}$  is a KdV
solution and
\eqref{eqn:uF} shows that it can be found through products, and sums
of derivatives of these polynomials followed by division by a
real-valued polynomial and hence each component
function is a rational function of $x$ and $t$.
\end{proof}

\begin{example}\label{examp:ratl2} The first example of a
quaternion-valued KdV solution above in Example~\ref{examp:ratl} was
an instance of this construction in the case $n=1$, $\alpha_0=\qk$ and
$\alpha_1=\qi$.
\end{example}

\begin{remark}\label{rem:Serna}
In fact, one may choose any finite linear
combination of the polynomials $\Delta_m(x,t)$ and create a
KdV-Darboux kernel out of this polynomial and all of its non-zero,
even order derivatives in $x$.  However, due to
Lemmas~\ref{lem:samespan} and \ref{lem:dropone}, no new KdV solutions
would be gained by considering even degree polynomials in $\F$ or including
lower order odd degree terms in the formula for $\phi_i$.  (See \cite{Serna} for details.)
\end{remark}

\section{Unions of KdV-Darboux Kernels}
 
Since the dispersion and closure properties of Definition~\ref{def:KdVDarbouxKernel} are
preserved under the taking of unions, it follows immediately
that:
\begin{theorem}
If $\F_1$ and $\F_2$ are KdV-Darboux kernels then $\F=\F_1\cup\F_2$ is
also a KdV-Darboux kernel as long as its Wronskian matrix is invertible.
\end{theorem}

This way of combining KdV-Darboux kernels allows for the creation of
$n$-soliton solutions or hybrids that exhibit features of more than
one of the basic solution types described in the previous section.
For example, there is a rational-periodic hybrid solution coming from
the union of the KdV-Darboux kernels in Examples~\ref{examp:ratl} and
\ref{examp:periodic}.  

The main focus of this section will be
 the case in which one additional function of the form
$\phi_{\alpha,\beta,\lambda}$ with $\lambda_0>0$ is added to a
KdV-Darboux kernel.  It is a consequence of Lemma~\ref{lem:dropone}
that the resulting solution will look like two
different solutions ``glued'' together, one visible to the left of the
localized disturbance that has been added and other other to the right
of it.  This general fact will be
both proved and illustrated in Section~\ref{sec:LR}.  Then in
Section~\ref{sec:2-sol} it will be used to derive a formula for
the phase shift of an arbitrary $2$-soliton solution.

\subsection{Asymptotics to the Left and Right of a Localized Disturbance}\label{sec:LR}
\begin{Proposition}\label{prop:LR}
Let $\F=\{\phi_1,\ldots,\phi_n\}$ be a KdV-Darboux kernel
where $n\geq2$ and $\phi_n$ has the form
$$
\phi_n(x,t)=\phi_{\alpha,\beta,\lambda}(x,t)=\alpha e^{\lambda x+
  \lambda^3 t}+\beta e^{-\lambda x - \lambda^3 t}
  $$
with 
$\lambda=\lambda_0+\lambda_1\qi,\ \lambda_0>0$.
 Then for each fixed $t$ and far enough to the left
of $x=c_{\alpha,\beta,\lambda}(t)$ the graph of $u_{\F}(x,t)$ looks
like\footnote{ We are being intentionally vague about what it means for one
solution to ``look like another'' to the right or left of some value
of $x$ here because the notation and proofs
both become unwieldy otherwise.  A rigorous mathematical definition
might include an arbitrarily small maximum amplitude for the difference of the two
functions when $x$ is is more than a certain distance to the right or
left of $c_{\alpha,\beta,\lambda}(t)$.  It is hoped that
Examples~\ref{examp:LRratl} and \ref{examp:onesidesing} illustrate both the meaning and significance of Proposition~\ref{prop:LR}.}
 the graph of 
$u_{\F^L}(x,t)$ where
$$
\F^L=\{Q^L(\phi_1),\ldots,Q^L(\phi_{n-1})\}
$$ 
with 
$Q^L(f)=f_x+\beta\lambda\beta^{-1}f$.
Similarly, for each fixed $t$ and far enough to the right of
$x=c_{\alpha,\beta,\lambda}(t)$ the solution $u_{\F}(x,t)$ looks like
$u_{\F^R}(x,t)$ where
$$
\F^R=\{Q^R(\phi_1),\ldots,Q^R(\phi_{n-1})\}
$$
with $Q^R(f)=f_x-\alpha\lambda\alpha^{-1}f$.
\end{Proposition}
\begin{proof}
The ``center'' $c_{\alpha,\beta,\lambda}(t)$ is the value of $x$ for
which there is a balance between the magnitude of the two exponential
terms in $\phi_n$.  For $x$ much smaller than it,  $\alpha e^{\lambda x + \lambda^3 t}$ is negligibly
small.  For those values of $x$, the solution $u_{\F}(x,t)$ will
not look noticeably different than it would if $\phi_n$ was equal to $\beta e^{-\lambda x - \lambda^3
  t}$, which according to Lemma~\ref{lem:dropone} is precisely
$u_{\F^L}$.  Similarly, when $x$ is much larger than
$c_{\alpha,\beta,\lambda}(t)$ the term $\beta e^{-\lambda x -
  \lambda^3 t}$ is negligibly small and the solution would not look
noticeably different than it would if that term was not there, which
is $u_{\F^R}$ according to Lemma~\ref{lem:dropone}.
\end{proof}

\begin{example}\label{examp:LRratl} Consider the ``hybrid'' rational/soliton solution $u_{\F}$
that comes from the choice
$$
\F=\{x+3\qk,\phi_{\alpha,\beta,\lambda}\}\hbox{ with
}\alpha=1,\ \beta=\qj,\ \lambda=2+\qi.
$$
According to Proposition~\ref{prop:LR} the left side of this solution
should look like 
$$
u_{\F^L}=\frac{-50 x^2-40 x+444}{\left(5 x^2+4 x+46\right)^2}
+\frac{ (5 x+2)}{\left(5 x^2+4 x+46\right)^2}(4\qi+48\qj+36\qk)
$$
and the right side should like look 
$$
u_{\F^R}=\frac{-50 x^2+40 x+444}{\left(5 x^2-4 x+46\right)^2}
+\frac{(5 x-2)}{\left(5 x^2-4 x+46\right)^2}(4\qi-48\qj+36\qk)
$$
Each of these is a stationary ($t$-independent) quaternion-valued KdV solution.  They
are shown in the left-most and right-most images of
Figure~\ref{fig:LRratl} respectively.
The middle two images of that
figure show $u_{\F}$ at times $t=-5$ and $t=5$.  Then we can see that
$u_{\F}$ looks like $u_{\F^L}$ to the left of the incoming soliton and
looks like $u_{\F^R}$ to the right of it.
\end{example}
\begin{figure}
\centering

\hbox to \textwidth{%
\includegraphics[width=1.28in,height=1.28in]{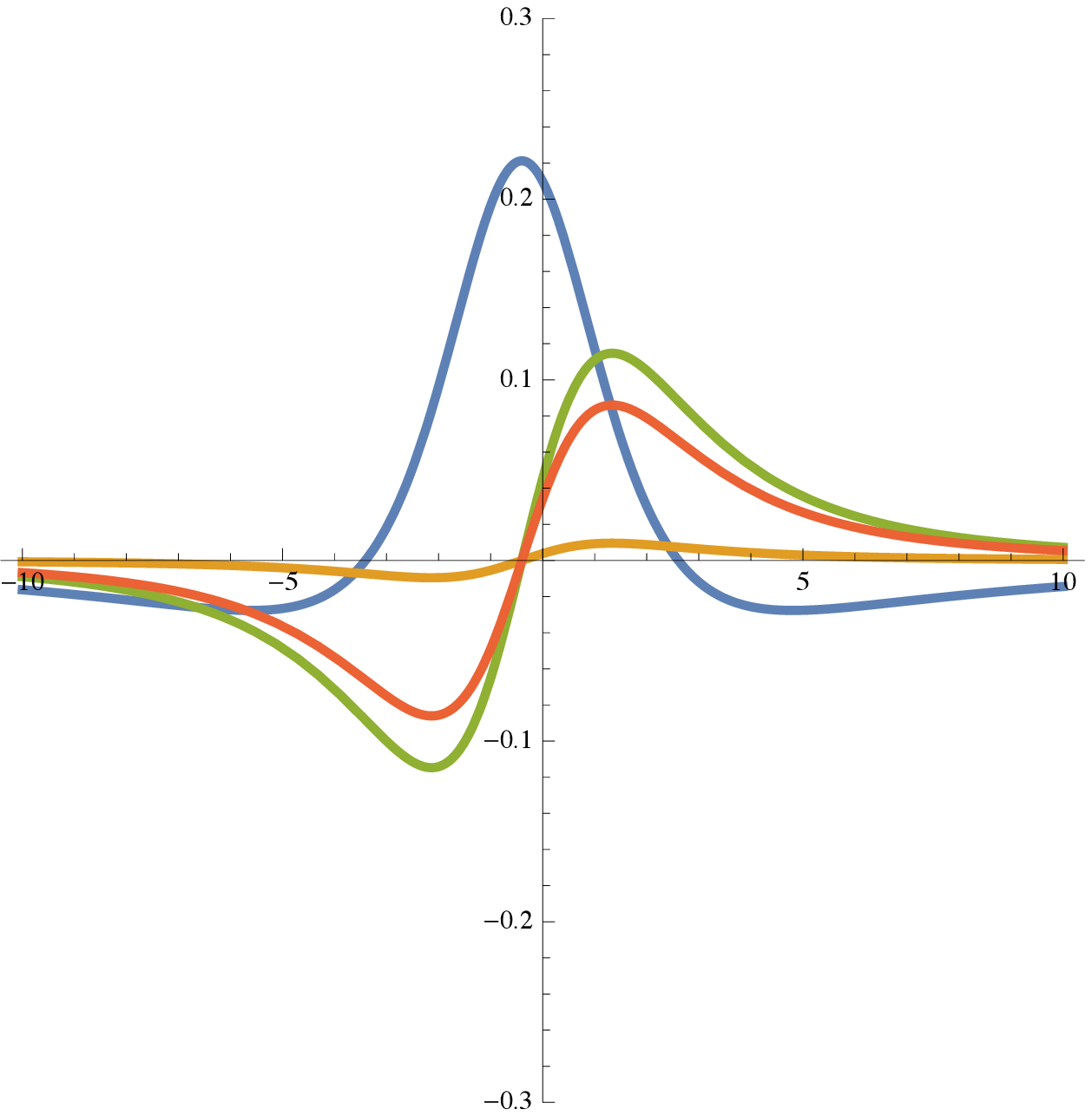}\hfill
\includegraphics[width=1.28in,height=1.28in]{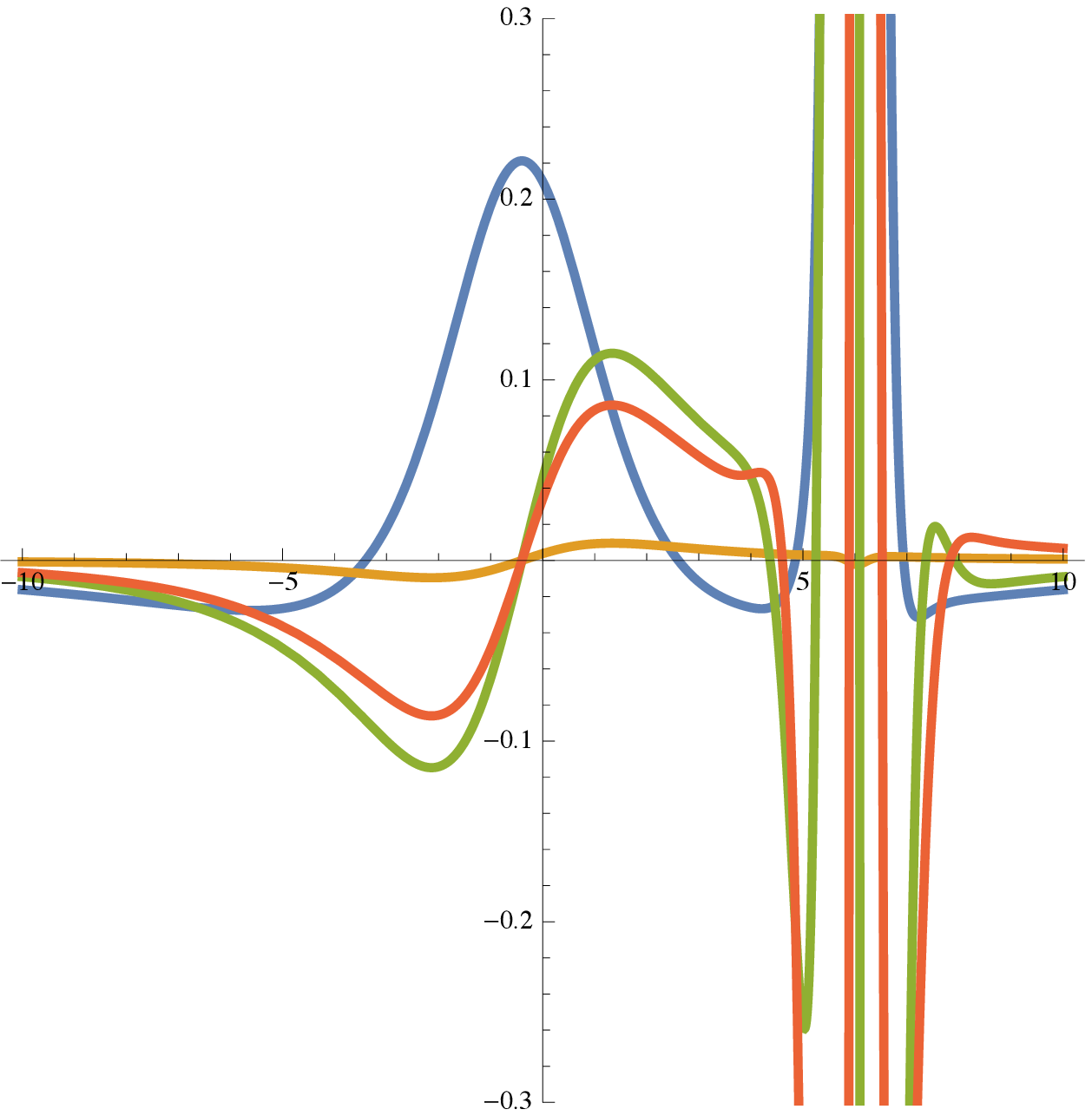}\hfill
\includegraphics[width=1.28in,height=1.28in]{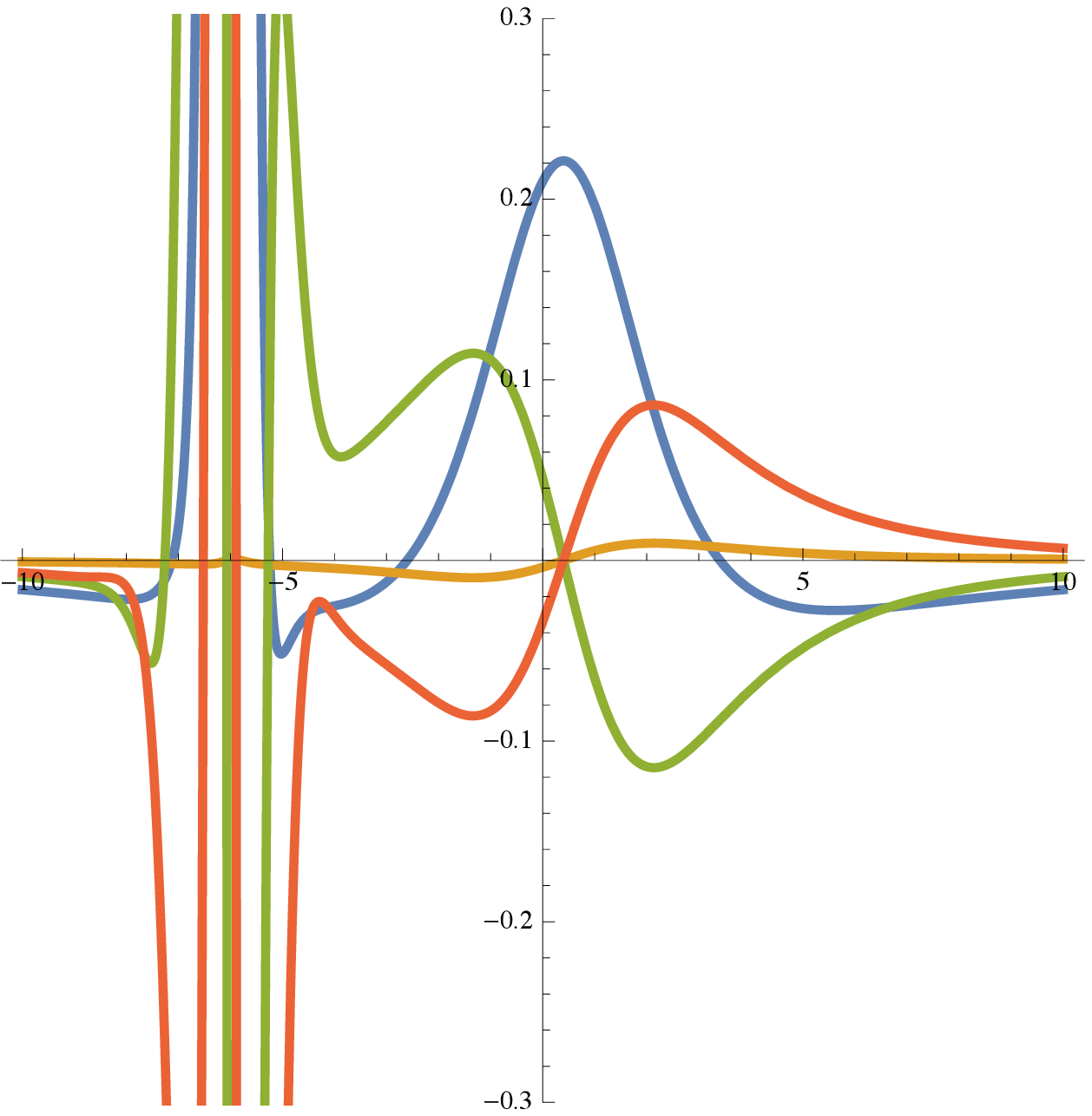}\hfill
\includegraphics[width=1.28in,height=1.28in]{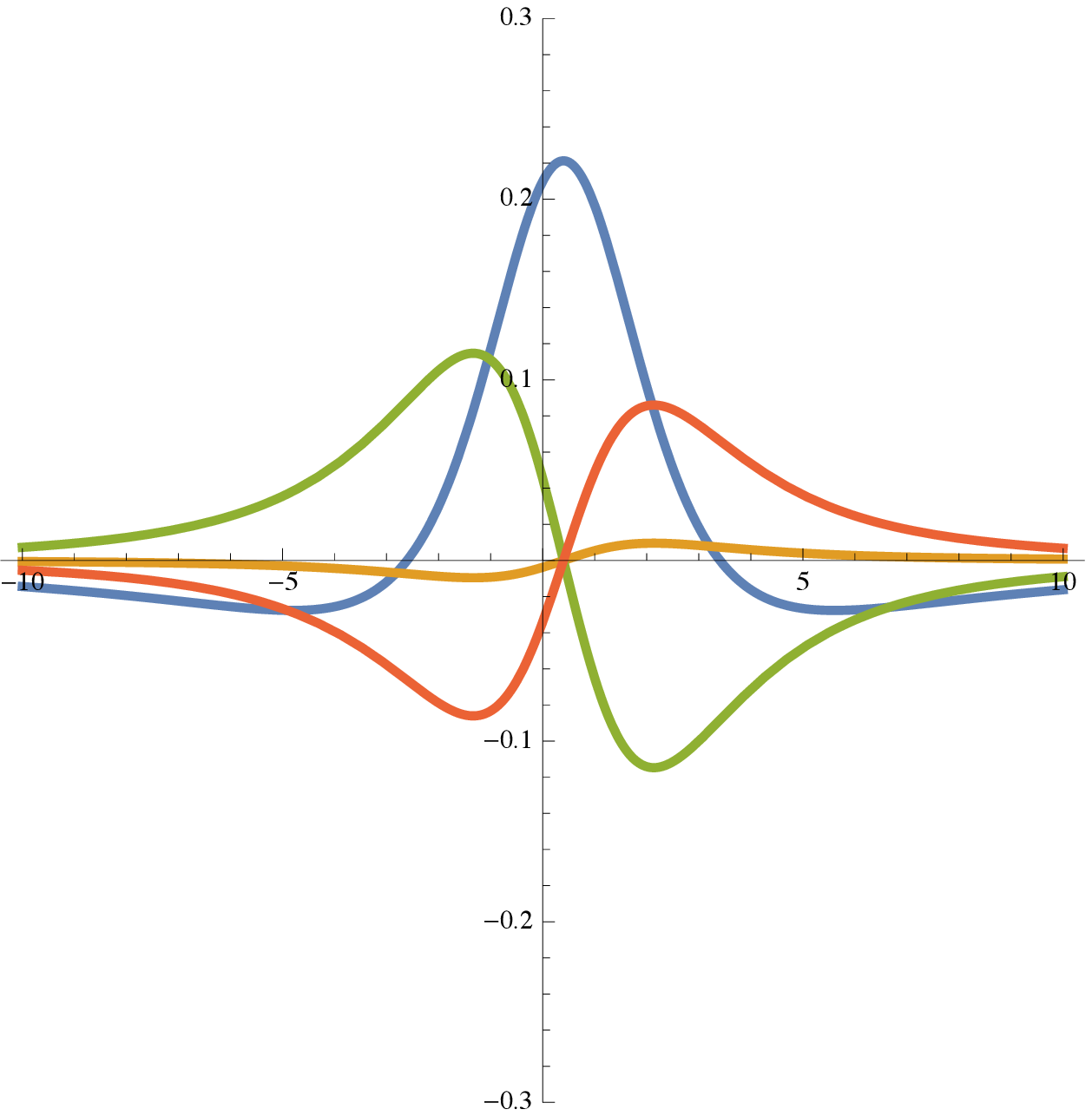}}

\caption{The figure on the left shows the stationary solution
  $u_{\F^L}$ and the figure on the right shows the stationary solution
  $u_{\F^R}$ from Example~\ref{examp:LRratl}.  The other two figures
  show the hybrid rational/soliton solution $u_{\F}$ at times $t=-5$
  and $t=5$ and one can see that it looks like $u_{\F^L}$ to the left
  of the localized disturbance and looks like $u_{\F^R}$ to the right
  of it.  Since a graph of $u_{\F}$ for very negative and positive
  times in the viewing window shown above would look indistinguishable
  from the figures at the far left and right above, one could look at
  them successively as representing an animation of the evolution in
  time of the solution $u_{\F}$: it begins looking like $u_{\F^L}$,
  then a localized disturbance comes in from the right and after it
  passes the solution now looks instead like $u_{\F^R}$.
}\label{fig:LRratl}
\end{figure}

\begin{remark}\label{rem:phitophi}
Since Proposition~\ref{prop:LR} will mostly be
applied in the case that each $\phi_i\in\F$ is a function of the form
\eqref{eqn:phiabl}, it is  worth noting that the
differential operator $Q$ defined by $Q(f)=f_x-\gamma f$ preserves that form.  In particular, it can
easily be computed that for any $\hatalpha,\hatbeta,\hatlambda\in\H$
\begin{equation}
Q(\phi_{\hatalpha,\hatbeta,\hatlambda})=\phi_{\tildealpha,\tildebeta,\hatlambda}\ 
\hbox{where}\ \tildealpha=\hatalpha\hatlambda-\gamma\hatalpha,\ \tildebeta=-\hatbeta\hatlambda-\gamma\hatbeta.
\label{eqn:Qphi}
\end{equation}
\end{remark}

\begin{example}\label{examp:onesidesing} Let $\F_1$ be the KdV-Darboux kernel
$$\F_1=\{\phi_{\alpha_1,\beta_1,\lambda_1}\}
\qquad\hbox{where}\qquad
\alpha_1=1,\ \beta_1=\frac{1}{\sqrt{5}},\ \lambda_1=\qi
$$
then $u_{\F_1}$ is a complex-valued, non-singular, periodic solution to
\eqref{eqn:KdV}.
And let $\F_2$ be the KdV-Darboux kernel
$$
\F_2=\{\phi_{\alpha_2,\beta_2,\lambda_2}\}\qquad
\hbox{where}\qquad
 \alpha_2=\qi+\qj,\ \beta_2=\qj+\qk,\ \lambda_2=1+\qi
$$
so that $u_{\F_2}$ is a breather soliton solution traveling to the
right at speed $2$.
What will the solution coming from $\F=\F_1\cup\F_2$ look like?
According to Proposition~\ref{prop:LR} and Remark~\ref{rem:phitophi}, to the left of $x=c_{\alpha_2,\beta_2,\lambda_2}=2t$ should
look like $u_{\alpha_L,\beta_L,\lambda_1}$ with
$$
\alpha_L=1,\ 
\beta_L=\frac{1}{\sqrt{5}}-\frac{2}{\sqrt{5}}\qi
$$
 while
on the right it should look like $u_{\alpha_R,\beta_R,\lambda_1}$
where
$$
\alpha_R=-1+\qi-\qj,\ 
\beta_R=-\frac{1}{\sqrt{5}}-\frac{1}{\sqrt{5}}\qi-\frac{1}{\sqrt{5}}\qj.
$$
  The interesting thing about
this is that since $|\beta_L|=1\not=|\beta_R|$ the solution it looks like on the
left is
\textit{singular} while the one on the right is not.  Figure~\ref{fig:LR} shows
that the solutions appear as predicted.  Moreover, this solution is
very interesting to watch as an animation because the localized
disturbance traveling to the right seems to transform a non-singular
periodic solution into a singular one as it passes.
\end{example}

\begin{figure}
\centering

\hbox to \textwidth{
\relax{\includegraphics[width=1.54in,height=1.54in]{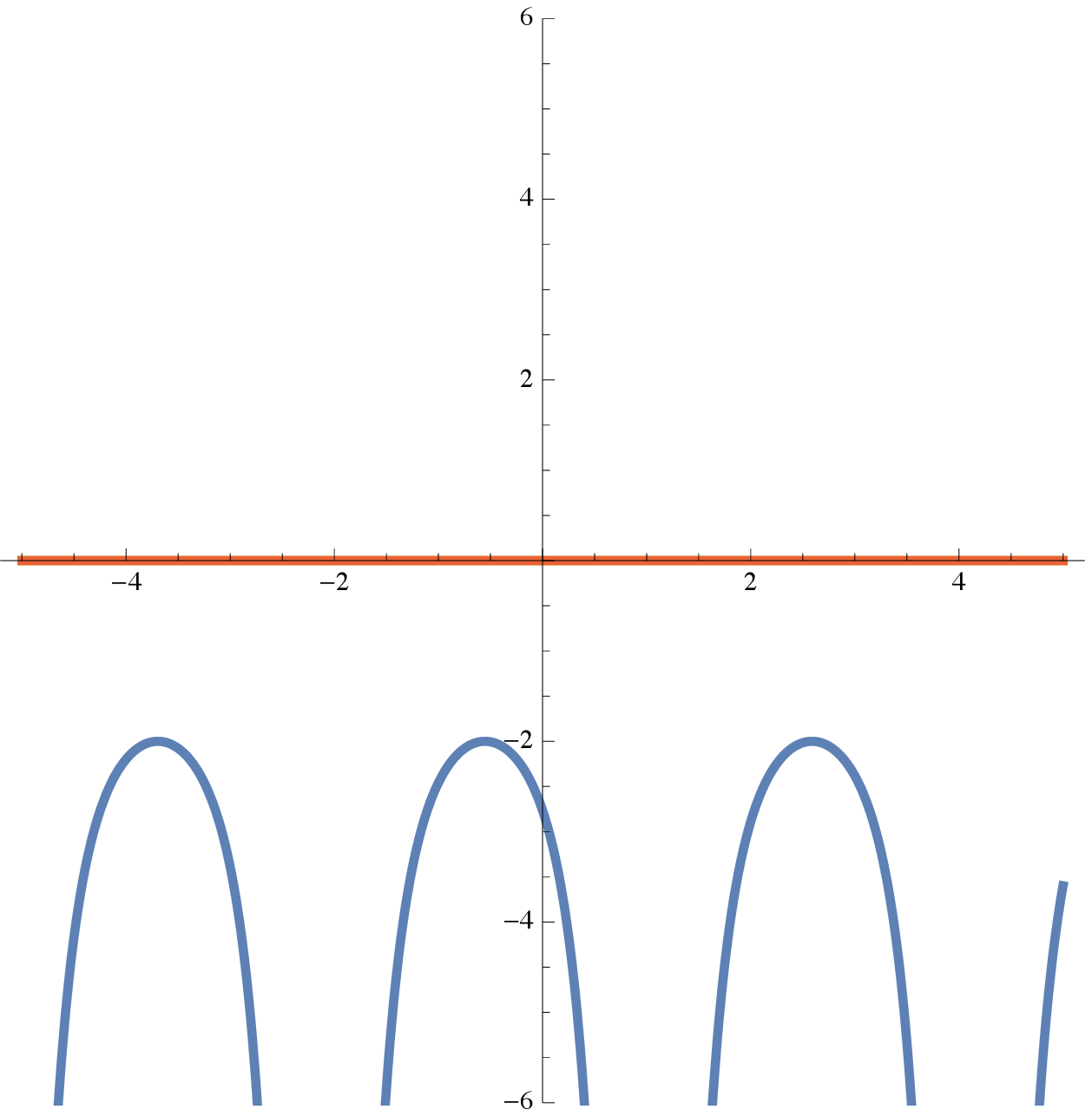}}\hfill
\relax{\includegraphics[width=1.54in,height=1.54in]{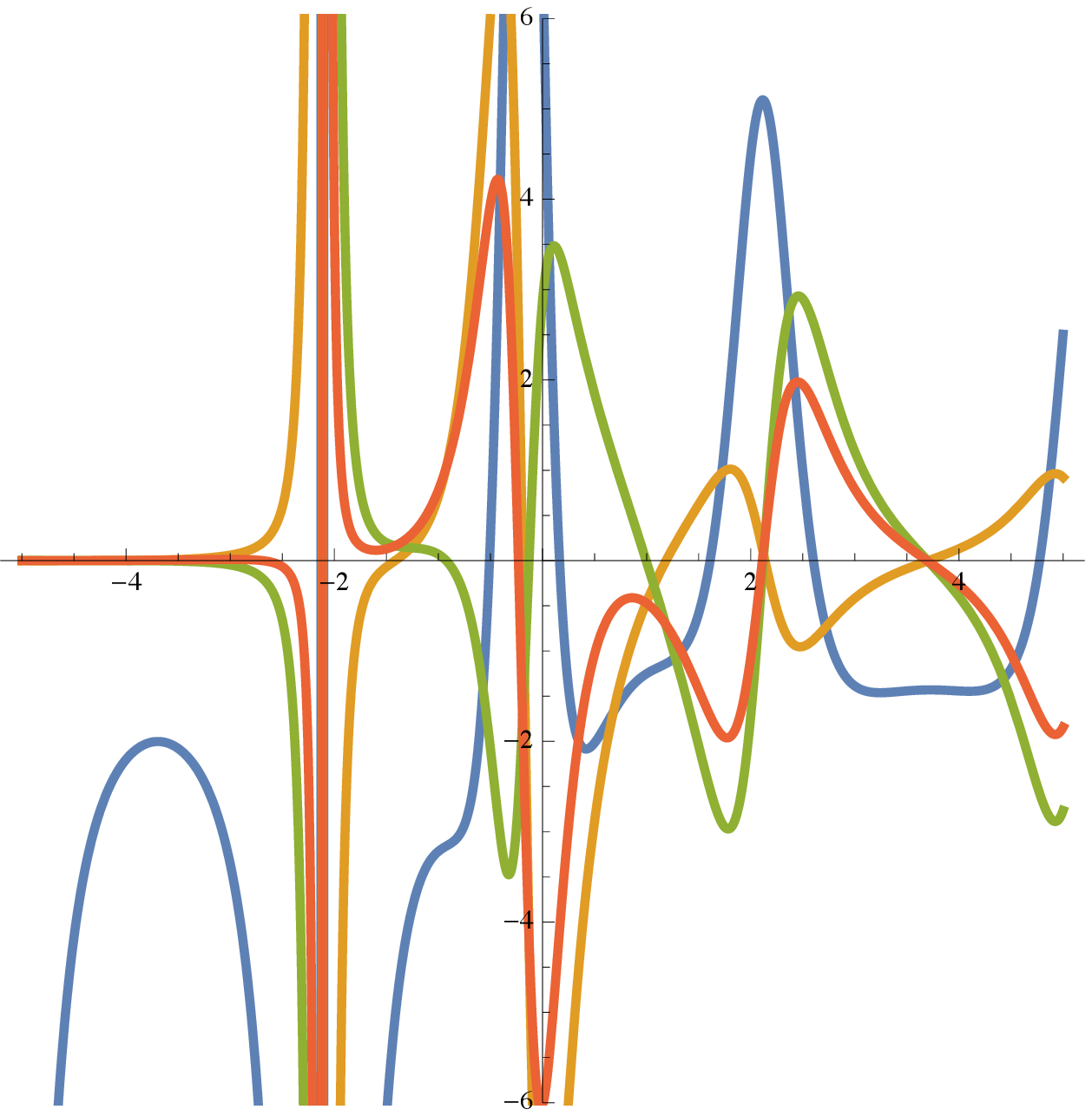}}\hfill
\relax{\includegraphics[width=1.54in,height=1.54in]{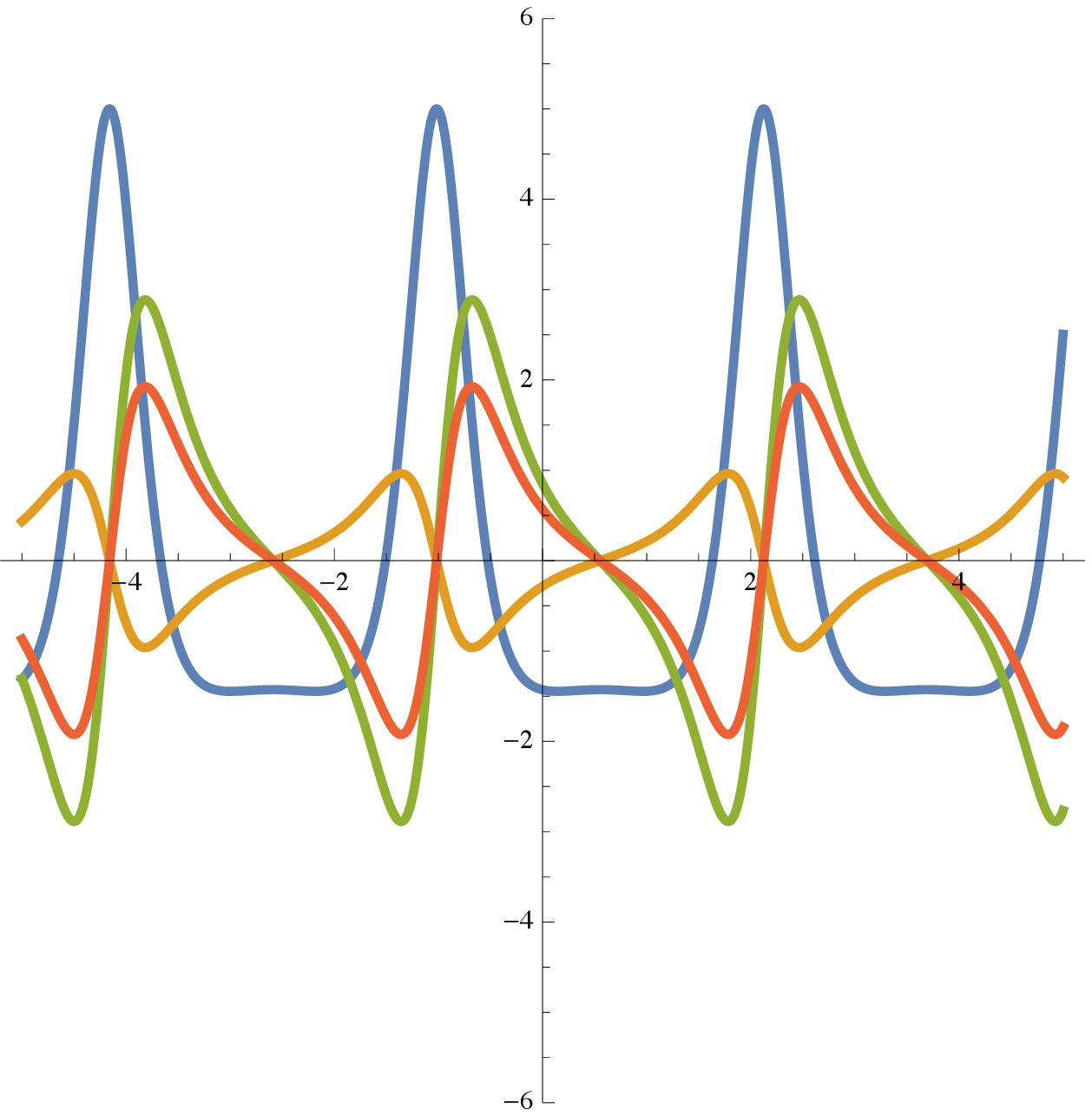}}
}

\caption{The quaternion-valued KdV solution shown in the middle is
  $u_{\F}$ from Example~\ref{examp:onesidesing}.  According to
  Proposition~\ref{prop:LR} it should look like the solution
  $u_{\alpha_L,\beta_L,\lambda_1}$ shown in the figure on the left for
 sufficiently negative values of $x$ and
 it should look like $u_{\alpha_R,\beta_R,\lambda_1}$,
  whose graph appears at the right for sufficiently positive values of
  $x$.  In fact, this convergence occurs quickly enough that one
  cannot visually tell the difference for $|x|>3$.
(All three solutions are shown at time $t=0$.)}\label{fig:LR}
\end{figure}

\subsection{Phase Shift of the General 2-soliton}\label{sec:2-sol}

Suppose $\alpha$, $\beta$, $\hat\alpha$ and $\hat\beta$ are non-zero
quaternions and that $\lambda$ and $\hat\lambda$ are complex numbers
such that:
$$
\lambda_0>0,\qquad\hat\lambda_0>0,\qquad\hbox{and}\qquad
\hat\lambda_0^2-3\hatlambda_1^2<\lambda_0^2-3\lambda_1^2.
$$
Then $u_{\alpha,\beta,\lambda}$ and
$u_{\hatalpha,\hatbeta,\hatlambda}$ are each quaternion-valued
1-soliton solutions to \eqref{eqn:KdV}.  Moreover, because of the last
inequality they have different velocities.  In particular, the center
$c_{\alpha,\beta,\lambda}(t)$ moves to the left more quickly than
$c_{\hatalpha,\hatbeta,\hatlambda}(t)$ (or, equivalently, moves to the
right more slowly).

This section will use Proposition~\ref{prop:LR} to analyze the
quaternion-valued KdV solution $u_{\F}$ where 
$$
\F=\{\phi_{\alpha,\beta,\lambda},\phi_{\hatalpha,\hatbeta,\hatlambda}\}.
$$
As will be seen below, 
there are constants  $\alpha_-,\beta_-,\hatalpha_-,\hatbeta_-\in\H$ so that
an observer watching an animation of $u_{\F}$ for sufficiently
negative values of $t$
 would see two localized disturbances traveling
with the same velocities as $c_{\alpha,\beta,\lambda}$ and
$c_{\hatalpha,\hatbeta\hatlambda}$.  In particular, there are quaternions
$\alpha_-,\beta_-,\hatalpha_-,\hatbeta_-\in\H$ so that the observer
would see what appeared to be the localized disturbances from the
solution $u_{\alpha_-,\beta_-,\lambda}$ and
$u_{\hatalpha_-,\hatbeta_-,\hatlambda}$ together in this $2$-soliton
solution.

There are also numbers $\alpha_+,\beta_+,\hatalpha_+,\hatbeta_+\in\H$
so that an animation of $u_{\F}$ for sufficiently positive values of
$t$ will look like a sum of the localized solutions
$u_{\alpha_+,\beta_+,\lambda}$ and
$u_{\hatalpha_+,\hatbeta_+,\hatlambda}$.  So, the observer would still
see two localized disturbances traveling with the same velocities.
However, they might not be in the locations that the observer would
have predicted.  Suppose the observer identifies the localized
disturbances of the same velocities at different times.  Then, since
at negative times the disturbance with greater velocity appears to be following the
linear trajectory $c_{\alpha_-,\beta_-,\lambda}(t)$, the observer may
expect it to be found on that same line at positive times as well.
Indeed, since $c_{\alpha_-,\beta_-,\lambda}$ and
$c_{\alpha_+,\beta_+,\lambda}(t)$ are moving to the left at the same
velocity, their difference is a constant.  If it is zero, then the
observer will find that the disturbance at later times is right where
it would have been expected to be.  However, if this difference is
non-zero, then it will be \textit{shifted} from its expected position.
Hence, the number
$c_{\alpha_+,\beta_+,\lambda}-c_{\alpha_-,\beta_-,\lambda}$ is called
the \textit{phase shift} of that localized disturbance and it is
generally interpreted as being a lasting consequence of its collision
with the other localized disturbance (though there are other
interpretations as well)\cite{ZK,BKY}. 

The following result provides formulas for the parameters for the
corresponding $1$-soliton solutions and the phase shifts of each of
the localized disturbances.   As the examples after it
demonstrate, unlike the real case, the phase shift of the solitary
wave with greater velocity  in the
quaternion-valued KdV $2$-soliton can be positive, negative, \textit{or} zero.  Additionally,
the solitary waves traveling at the same velocities at positive and
negative may differ in \textit{shape} as well as being horizontally
shifted relative to each other.

\begin{Proposition}\label{prop:PhaseShift}
Choose constant $\alpha$, $\beta$, $\lambda$, $\hatalpha$, $\hatbeta$,
and $\hatlambda$ as described above and let $\F$ be the KdV-Darboux kernel
$$\F=\{\phi_{\alpha,\beta,\lambda},\phi_{\hat\alpha,\hat\beta,\hatlambda}\}.$$
\begin{itemize}
\item For sufficiently negative values of $t$, the solution $u_{\F}$ will
look like the $1$-soliton $u_{\alpha_-,\beta_-,\lambda}$ near its
center $x=c_{\alpha_-,\beta_-,\lambda}(t)$ and will also look like the
$1$-soliton  $u_{\hatalpha_-,\hatbeta_-,\hatlambda}$ near its
center $x=c_{\hatalpha_-,\hatbeta_-,\hatlambda}(t)$
where 
$$
\alpha_-= \alpha\lambda-\hatalpha\hatlambda\hatalpha^{-1}\alpha,
\qquad
\beta_-=-\beta\lambda-\hatalpha\hatlambda\hatalpha^{-1}\beta,
$$
$$
\hatalpha_-= \hatalpha\hatlambda+\beta\lambda\beta^{-1}\hatalpha,
\qquad \hbox{and}\qquad 
\hatbeta_-=-\hatbeta\hatlambda+\beta\lambda\beta^{-1}\hatbeta.
$$

\item For sufficiently positive values of $t$, the solution $u_{\F}$ will
look like the $1$-soliton $u_{\alpha_+,\beta_+,\lambda}$ near its
center $x=c_{\alpha_+,\beta_+,\lambda}(t)$ and will also look like the
$1$-soliton  $u_{\hatalpha_+,\hatbeta_+,\hatlambda}$ near its
center $x=c_{\hatalpha_+,\hatbeta_+,\hatlambda}(t)$
where 
$$
\alpha_+=\alpha\lambda+\hatbeta\hatlambda\hatbeta^{-1}\alpha
\qquad
\beta_+=-\beta\lambda+\hatbeta\hatlambda\hatbeta^{-1}\beta
$$
$$
\hatalpha_+= \hatalpha\hatlambda-\alpha\lambda\alpha^{-1}\hatalpha
\qquad \hbox{and}\qquad 
\hatbeta_+=-\hatbeta\hatlambda-\alpha\lambda\alpha^{-1}\hatbeta
$$
\item  An observer looking at the solution $u_{\F}$ at very negative
   values of time would see two localized disturbances that appear to be moving
   at constant speeds.  Looking again at a very positive time the
   observer would see two localized disturbances with the same
   velocities, but they might not be in the locations that would have
   been expected if they had continued to follow a linear trajectory.
 The phase shifts experienced by the interacting solitary waves are
$$
c_{\alpha_+,\beta_+,\lambda}(t)-c_{\alpha_-,\beta_-,\lambda}(t)=\frac{\ln(\gamma)}{2\lambda_0}
\qquad\hbox{ and }\qquad
c_{\hatalpha_+,\hatbeta_+,\hatlambda}(t)-c_{\hatalpha_-,\hatbeta_-,\hatlambda}(t)=-\frac{\ln(\gamma)}{2\hatlambda_0}
$$
where
$$
\gamma=\frac{|\alpha\lambda\alpha^{-1}-\hatalpha\hatlambda\hatalpha^{-1}||
\beta\lambda\beta^{-1}-\hatbeta\hatlambda\hatbeta^{-1}|}{|\alpha\lambda\alpha^{-1}+\hatbeta\hatlambda\hatbeta^{-1}||\beta\lambda\beta^{-1}+\hatalpha\hatlambda\hatalpha^{-1}|}.$$
\end{itemize}
\end{Proposition}
\begin{proof}
First, we will apply 
Proposition~\ref{prop:LR} to $\F$ using
$\phi_{\hatalpha,\hatbeta,\hatlambda}$ in the role of $\phi_n$.
The proposition says that for all $t$, far enough to the right of $c_{\hatalpha,\hatbeta,\hatlambda}(t)$, the
graph of the solution $u_{\F}$ will look like
$u_{\alpha_-,\beta_-,\lambda}$ using the values from above.  Because
the center $c_{\alpha_-,\beta_-,\lambda}(t)$ is moving to the left more
quickly, for any sufficiently negative value of $t$, it will be far
enough to the right of $c_{\hatalpha,\hatbeta,\hatlambda}$ so that
it is in the region where $u_{\F}$ looks like
$u_{\alpha_-,\beta_-,\lambda}$.  Therefore for very negative values of
$t$, $u_{\F}$ has a localized disturbance that looks like the one in
the one soliton $u_{\alpha_-,\beta_-,\lambda}$.

On the other hand, when $t$ is very positive then
$c_{\alpha_+,\beta_+,\lambda}(t)$ will be far to the right of
$c_{\hatalpha,\hatbeta,\hatlambda}(t)$.  In particular
because it is moving to the
left at a faster constant speed when $t$ is sufficiently large it will
be located in the region where according to Proposition~\ref{prop:LR}
the solution $u_{\F}$ will look like $u_{\alpha_+,\beta_+,\lambda}$.
Since that is the location around which the soliton is localized in
that $1$-soliton, the solution $u_{\F}$ will look like that 1-soliton
near that point.

The other parameters come from repeating this same process but now
using $\phi_{\alpha,\beta,\lambda}$ in the role of $\phi_n$ when
applying Proposition~\ref{prop:LR}.

Now, $c_{\alpha_+,\beta_+,\lambda}(t)$ and
$c_{\alpha_-,\beta_-,\lambda}(t)$ are two linear trajectories with the
same velocity. The difference between them is the phase shift, how
much further to the right the disturbance is after the collision than
it would have been if it had continued to look like
$u_{\alpha_-,\beta_-,\lambda}$.  Using the formula for the center,
properties of logarithms and properties of the length operator on
quaternions, we can see that
\begin{eqnarray*}
c_{\alpha_+,\beta_+,\lambda}(t)-c_{\alpha_-,\beta_-,\lambda}(t)
&=&\frac{\ln|\alpha_+^{-1}\beta_+|}{2\lambda_0}-\frac{\ln|\alpha_-^{-1}\beta_-|}{2\lambda_0}
=\frac{1}{2\lambda_0}\ln\left(\frac{|\alpha_+^{-1}\beta_+|}{|\alpha_-^{-1}\beta_-|}\right)\\
&=&\frac{1}{2\lambda_0}\ln\left(\frac{|\alpha_-||\beta_+|}{|\alpha_+||\beta_-|}\right)\\
 &=& \frac{1}{2\lambda_0}\ln\left(\frac{|\alpha\lambda-\hatalpha\hatlambda\hatalpha^{-1}\alpha||
-\beta\lambda+\hatbeta\hatlambda\hatbeta^{-1}\beta|}{|\alpha\lambda+\hatbeta\hatlambda\hatbeta^{-1}\alpha||-\beta\lambda-\hatalpha\hatlambda\hatalpha^{-1}\beta|}\right)\\
&=&\frac{1}{2\lambda_0}\ln\left(\frac{|\alpha\lambda\alpha^{-1}-\hatalpha\hatlambda\hatalpha^{-1}||
\beta\lambda\beta^{-1}-\hatbeta\hatlambda\hatbeta^{-1}|}{|\alpha\lambda\alpha^{-1}+\hatbeta\hatlambda\hatbeta^{-1}||\beta\lambda\beta^{-1}+\hatalpha\hatlambda\hatalpha^{-1}|}\right)=\frac{\ln(\gamma)}{2\lambda_0}.
\end{eqnarray*}
A similar calculation for
$c_{\hatalpha_+,\hatbeta_+,\hatlambda}-c_{\hatalpha_-,\hatbeta_-,\hatlambda}$
results in the same formula but with the numerator and denominator
switched in the argument of the logarithm with the effect of changing
the sign.
\end{proof}
\begin{figure*}
\hbox to \textwidth{%
\includegraphics[width=1.32in,height=1.32in]{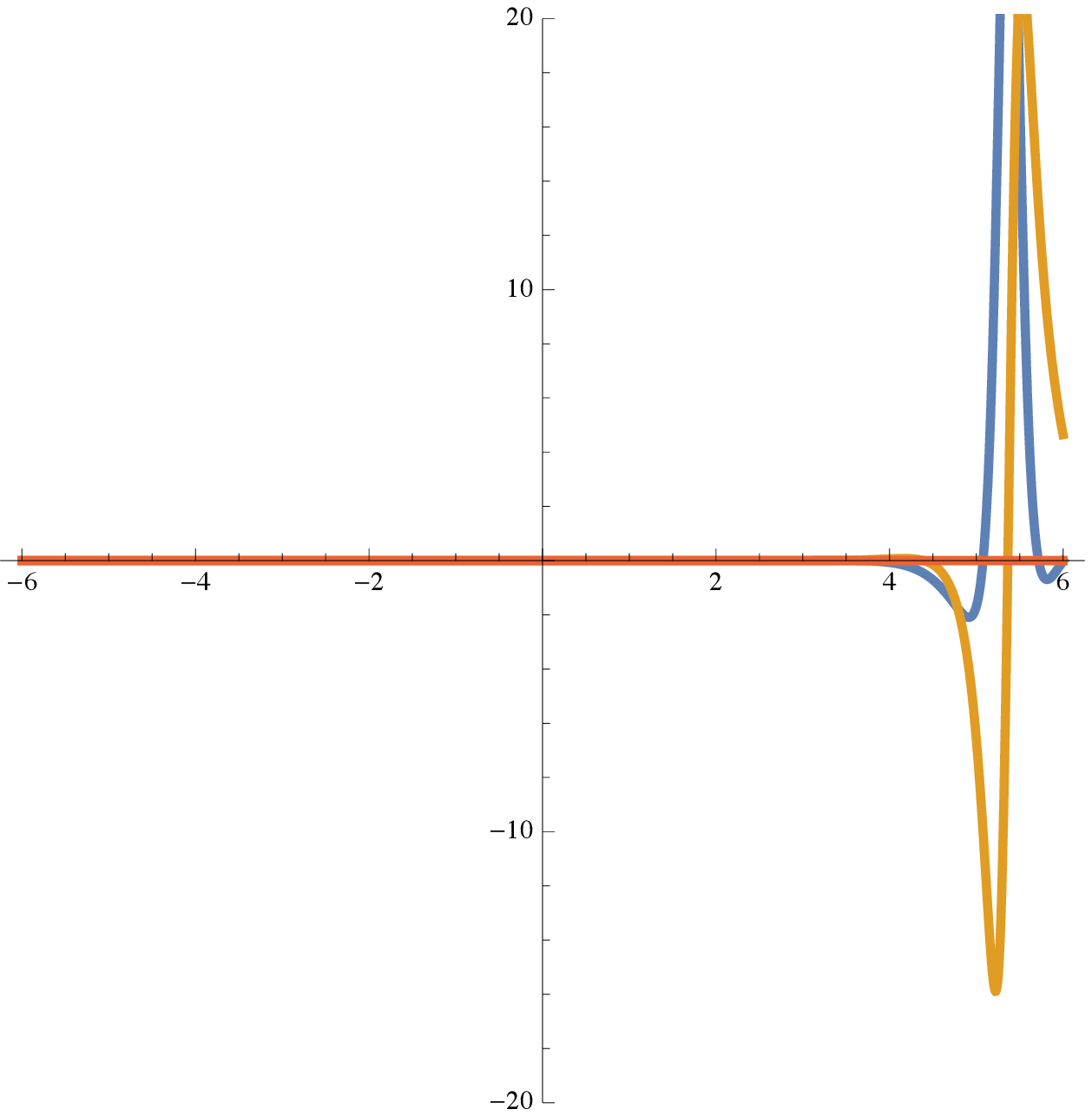}\hfill
\includegraphics[width=1.32in,height=1.32in]{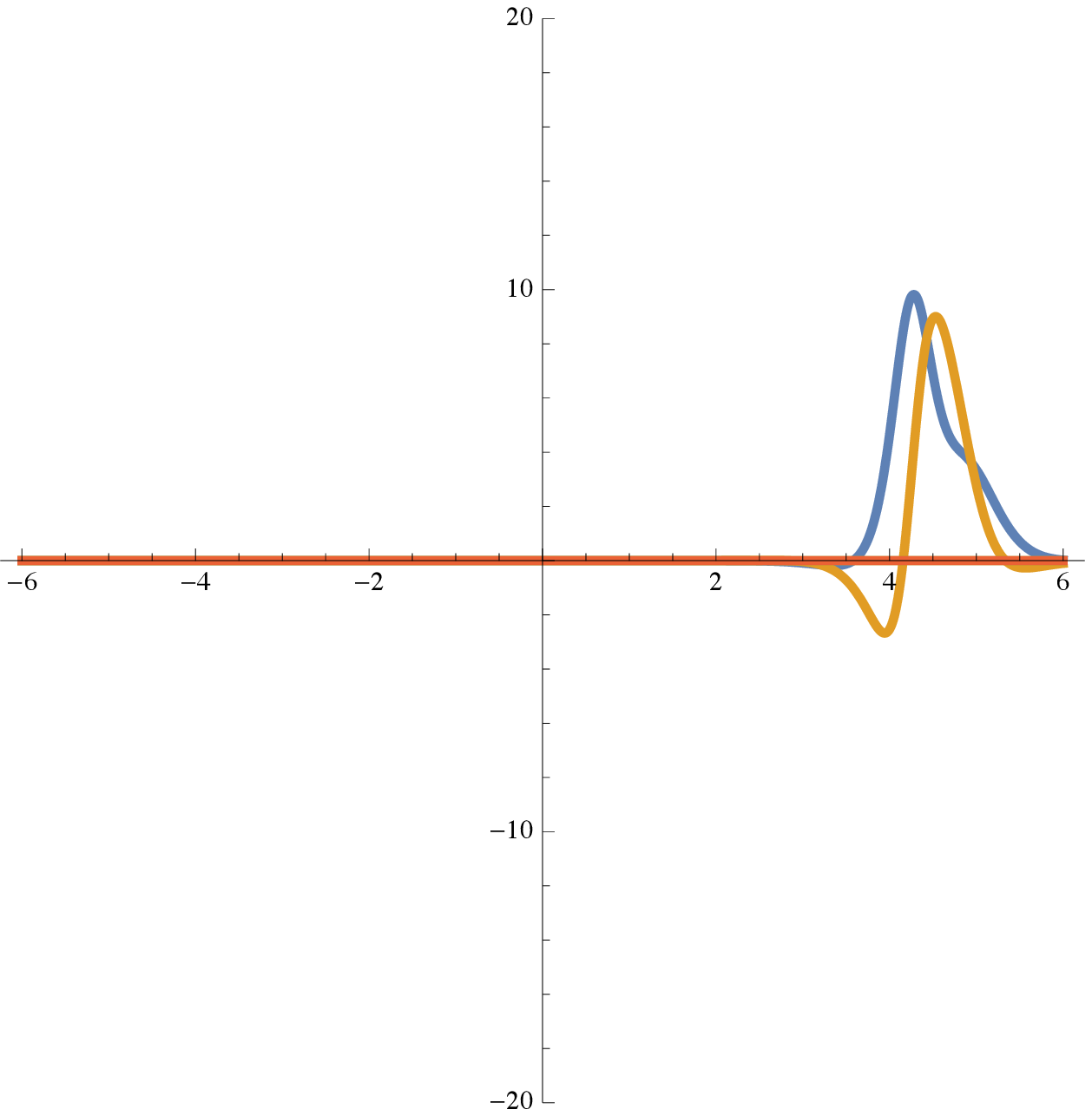}\hfill
\includegraphics[width=1.32in,height=1.32in]{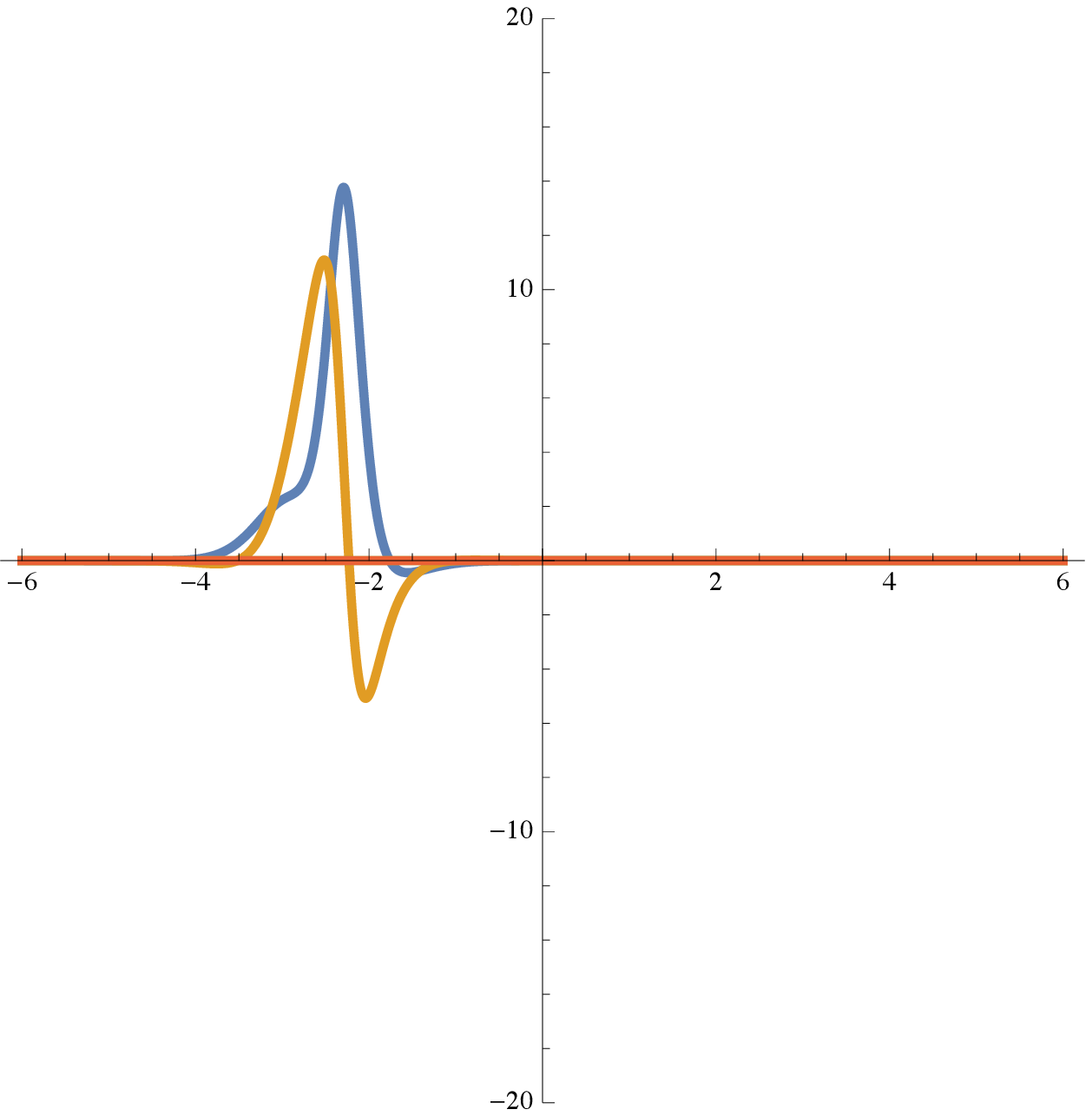}\hfill
\includegraphics[width=1.32in,height=1.32in]{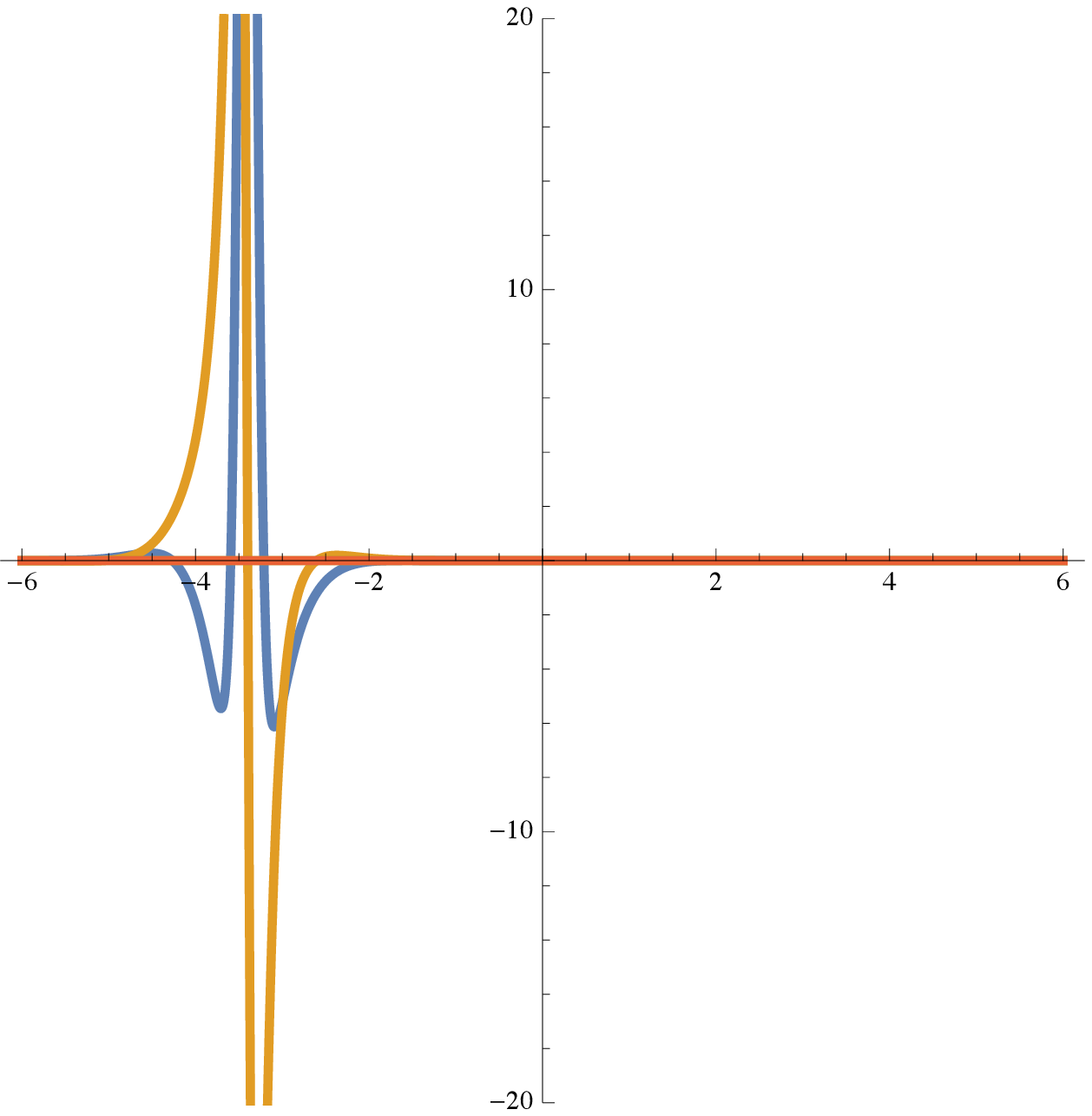}}

\hbox to \textwidth{%
\includegraphics[width=1.32in,height=1.32in]{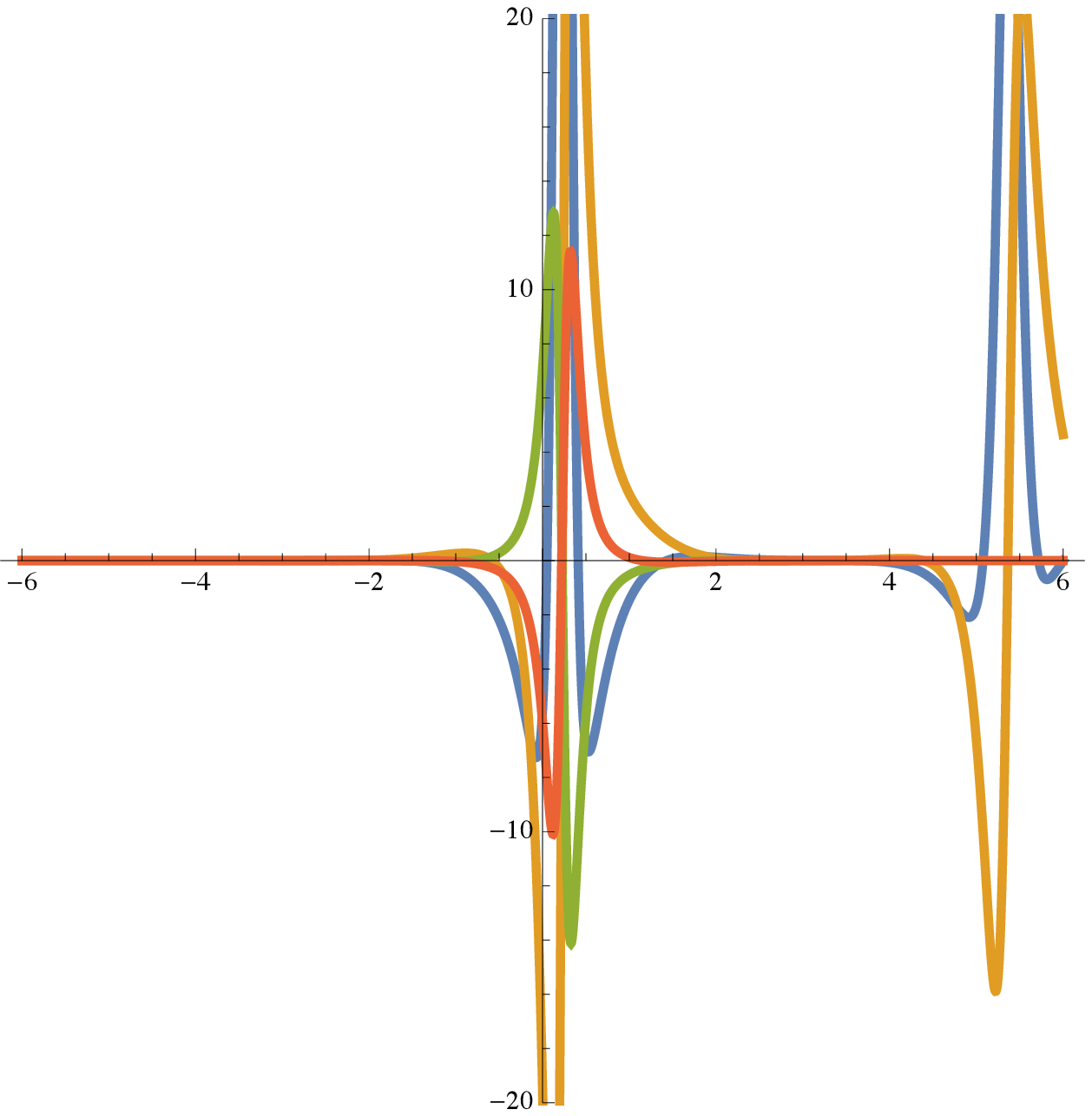}\hfill
\includegraphics[width=1.32in,height=1.32in]{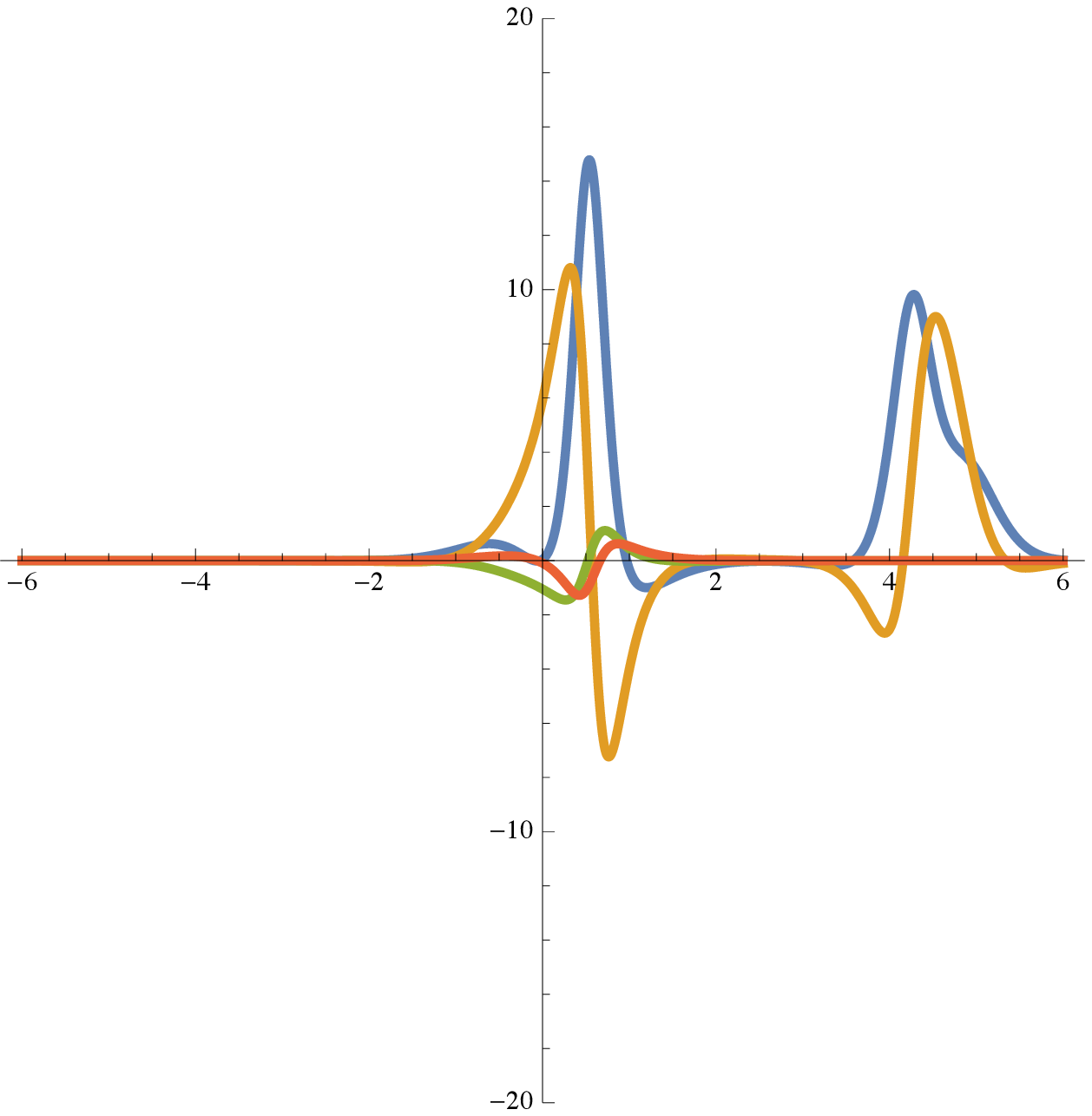}\hfill
\includegraphics[width=1.32in,height=1.32in]{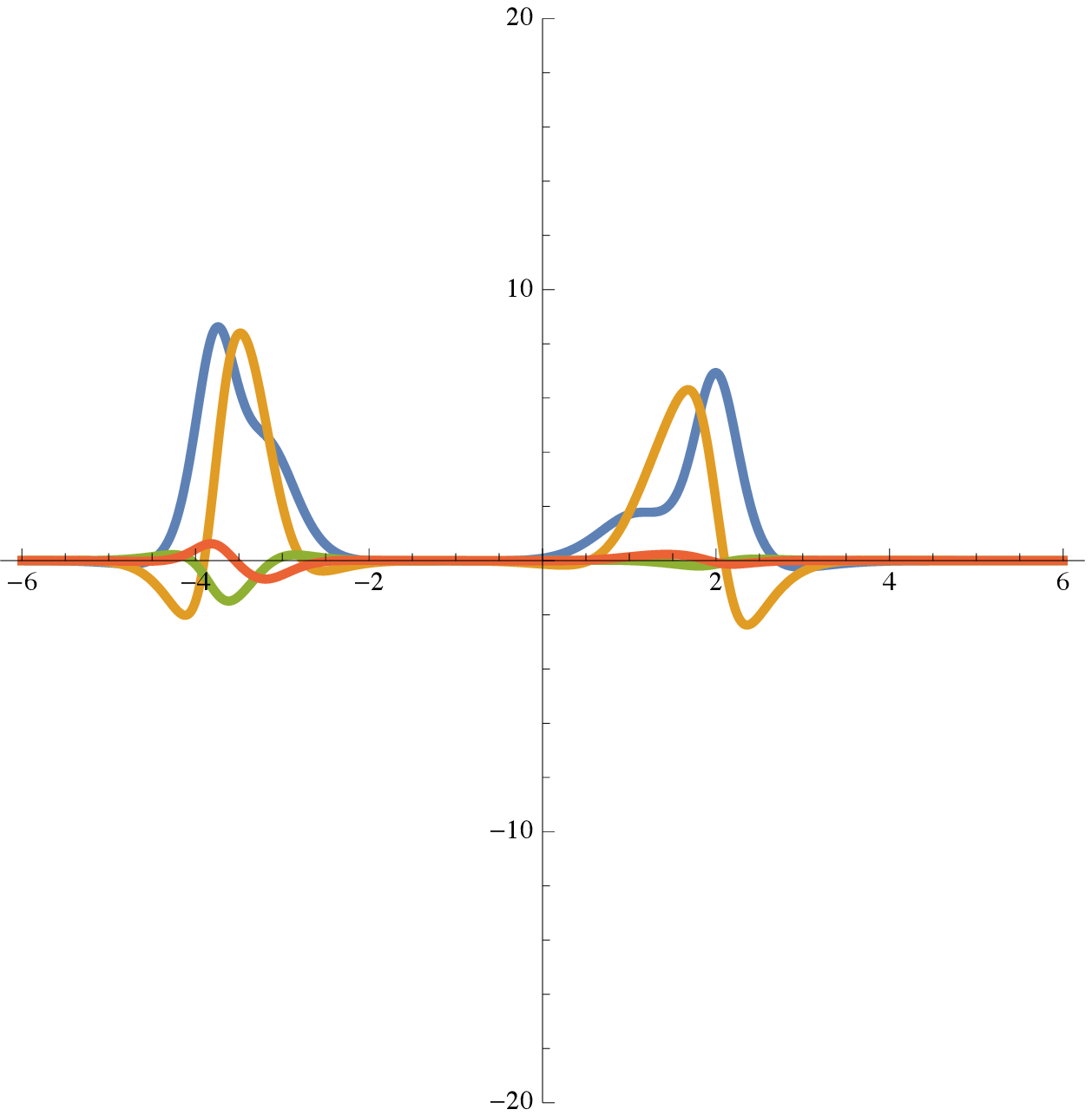}\hfill
\includegraphics[width=1.32in,height=1.32in]{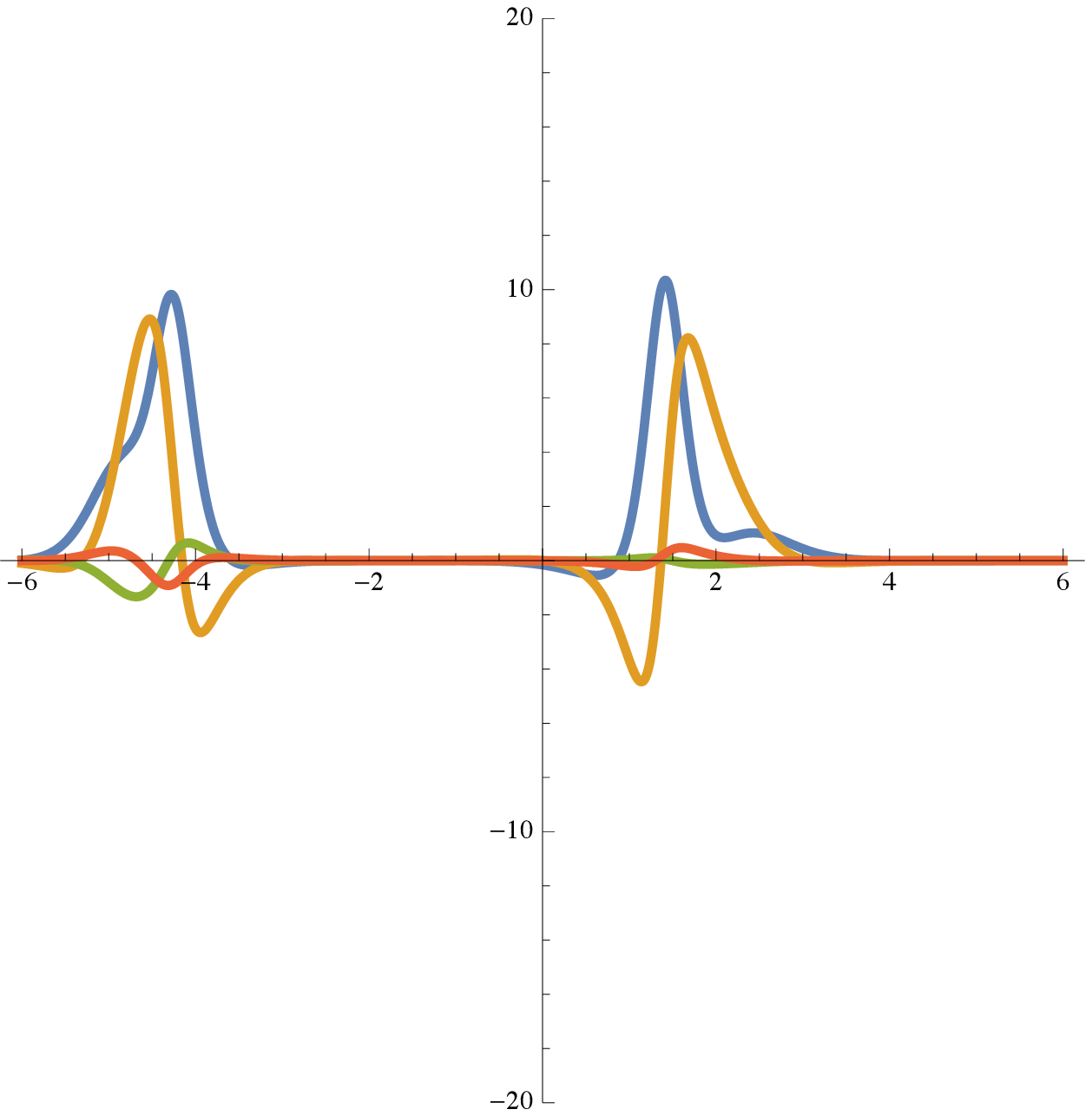}}

\hbox to \textwidth{%
\includegraphics[width=1.32in,height=1.32in]{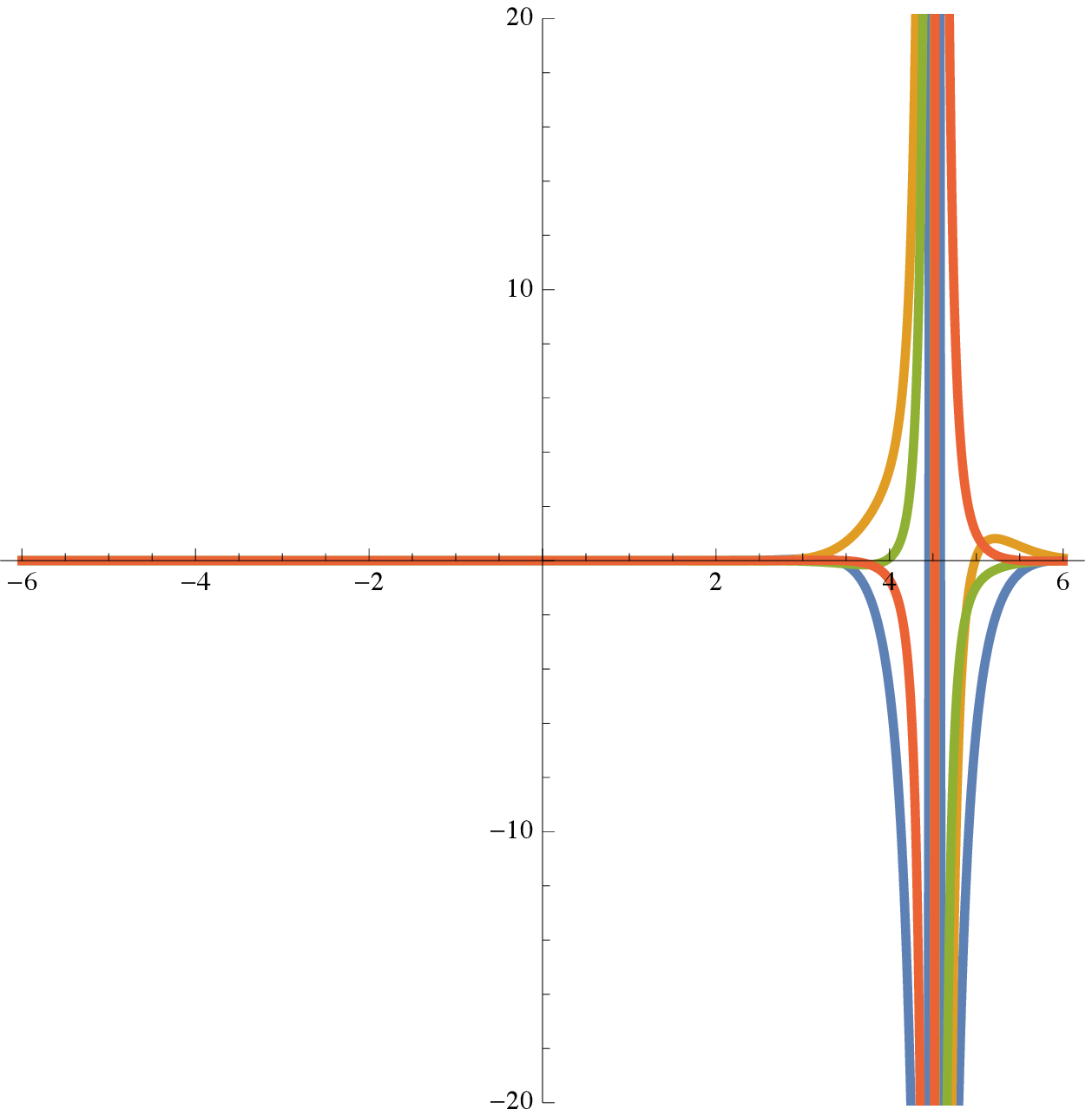}\hfill
\includegraphics[width=1.32in,height=1.32in]{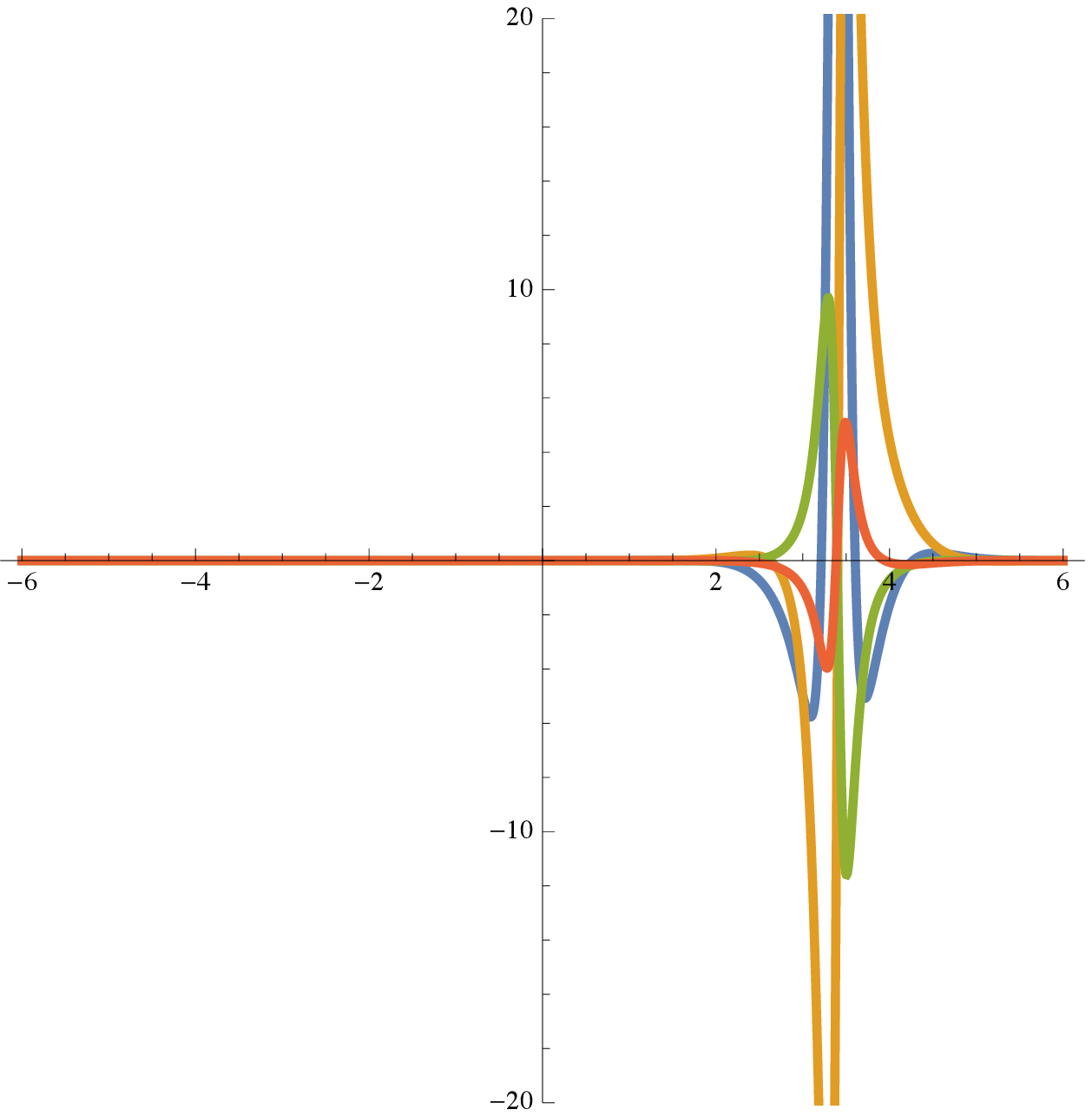}\hfill
\includegraphics[width=1.32in,height=1.32in]{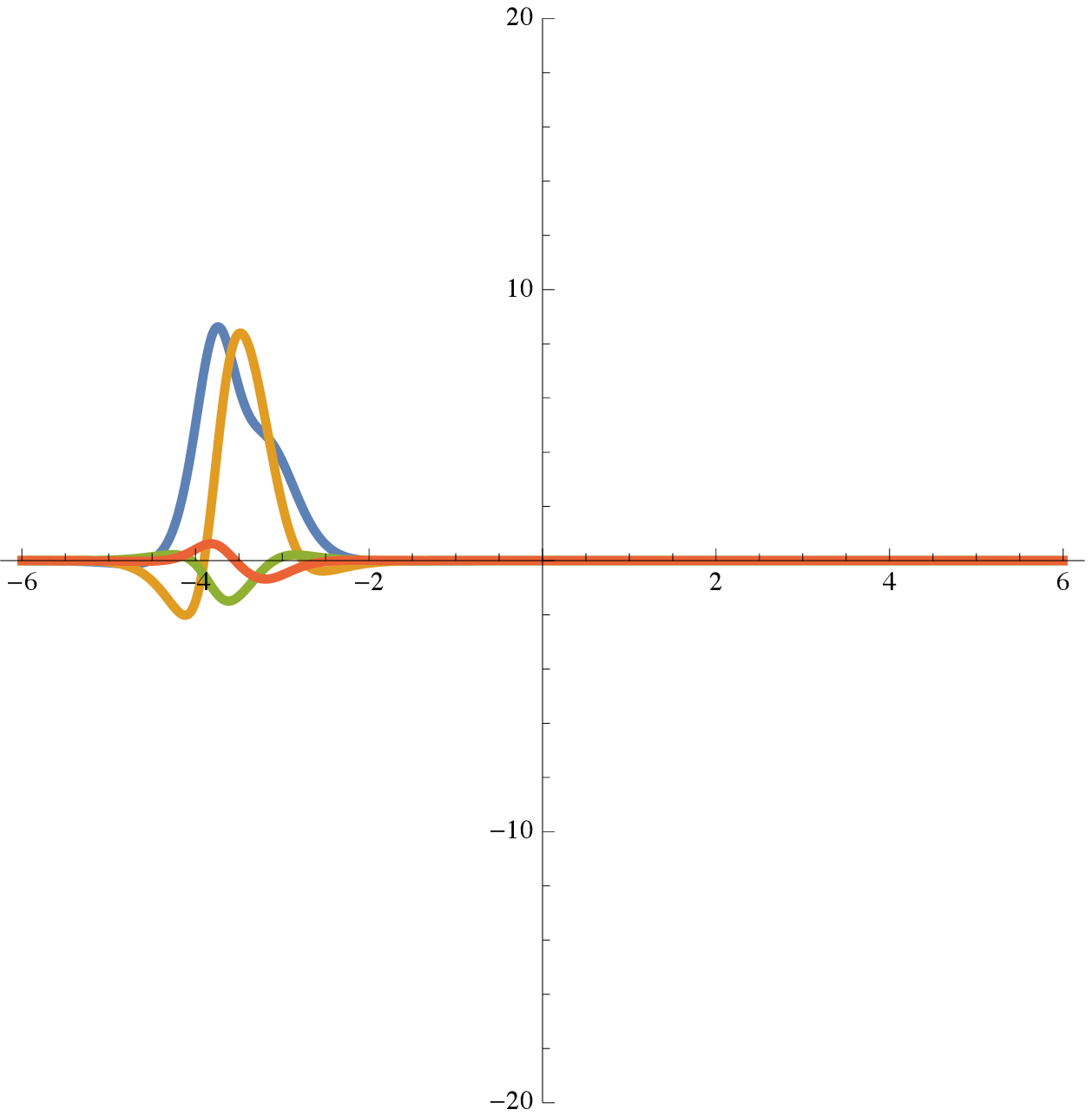}\hfill
\includegraphics[width=1.32in,height=1.32in]{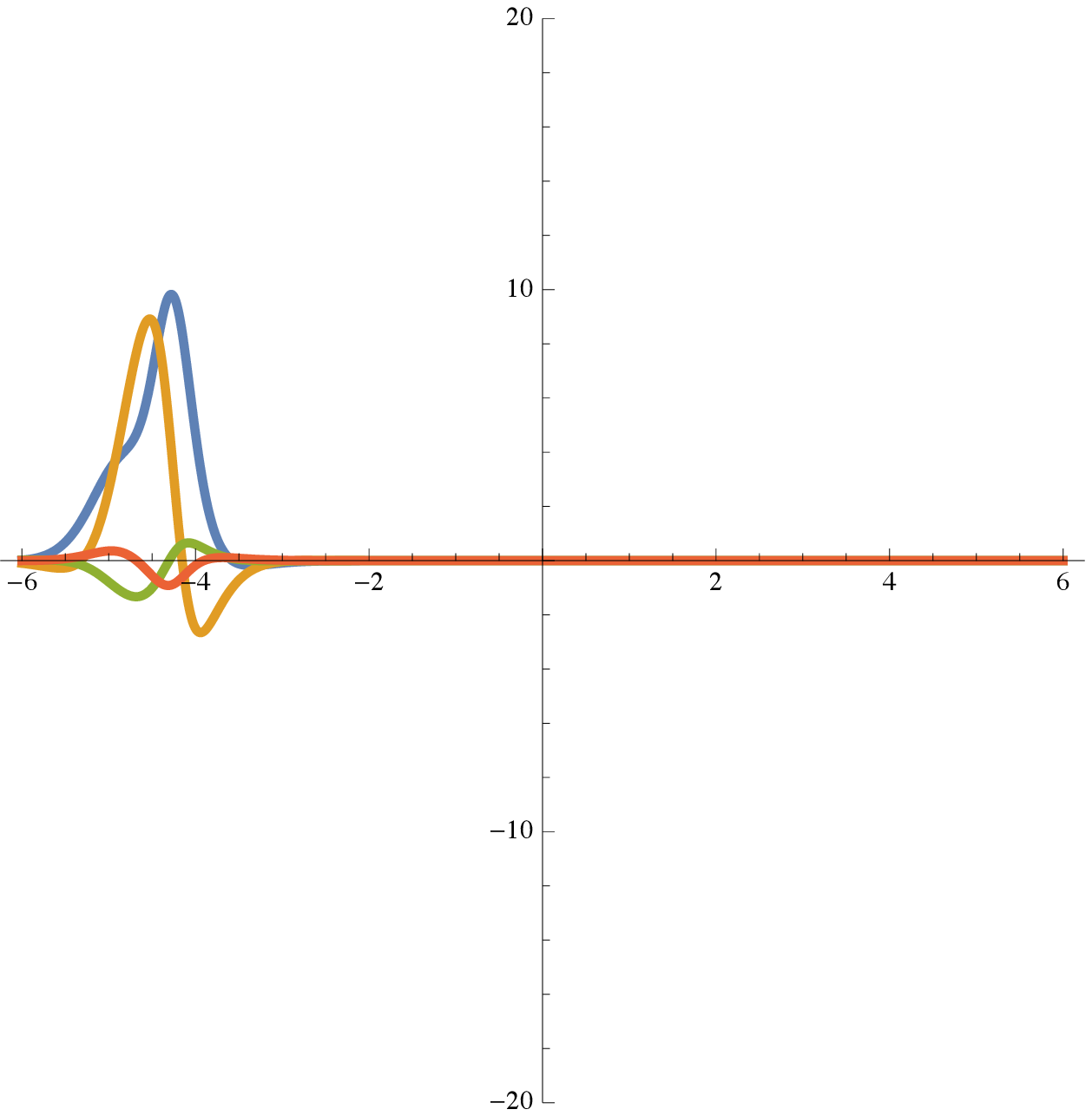}}

\caption{
The middle row of pictures shows the $2$-soliton $u_{\F}(x,t)$ from
  Example~\ref{examp:phaseshift} at times $t=-5$, $t=-4$, $t=3$ and
  $t=4$.  The first row shows the 1-soliton $u_{\alpha_-,\beta_-,\lambda}(x,t)$ and
  the third row shows the 1-soliton $u_{\alpha_+,\beta_+,\lambda}$ at
  the same times.  Notice that at negative times the moving solitary
  wave in the middle row is visually indistinguishable from the one
  above it.  However, at later times, it looks instead like the one
  below it.
Each of the 1-soliton solutions is moving to the left with the same
velocity and so the same horizontal shift relates their centers
at any time.  This is the ``phase shift'', which is
traditionally thought of as being an effect that the interaction of
$u_{\alpha_-,\beta_-,\lambda}$ had upon collision with the stationary
soliton.
  (Notice also that the stationary soliton has
  experienced a phase shift in the opposite direction.  It is further
  to the right after the collision than it was before.)}\label{fig:phaseshift}
\end{figure*}

\begin{example}\label{examp:phaseshift}
Imagine a naive observer watching an animation of the 2-soliton solution $u_{\F}$
where
$\F=\{\phi_{\alpha,\beta,\lambda},\phi_{\hatalpha,\hatbeta,\hatlambda}\}$
with
$$\alpha=\beta=\hatalpha=1,
\ 
\hatbeta=20\qi+\qk,
\ 
\lambda=2+\qi,
\ 
\hatlambda=3/2+\sqrt{3}/2\qi.
$$
Watching for some negative values of $t$, the observer would see a stationary
breather soliton sitting just a bit to the right of $x=0$ and then
another breather soliton approaching from the right at speed $1$.  In
fact, as predicted by Proposition~\ref{prop:PhaseShift} the moving
soliton looks like $u_{\alpha_-,\beta_-,\lambda}$ as shown 
at the top-left of Figure~\ref{fig:phaseshift}. In particular, for
these negative times, the moving localized disturbance in $u_{\F}$ is
located at
$$
x=c_{\alpha_-,\beta_-,\lambda}(t)=\frac{\log \left(31+16 \sqrt{3}\right)}{8}-t.
$$
So, the naive observer might assume that this will
continue to be true in the future.  However, as the images on the
right in Figure~\ref{fig:phaseshift} show, it does not.  Although
there is still a disturbance moving left at speed $1$ at times $t=3$
and $t=4$, it no longer looks like $u_{\alpha_-,\beta_-,\lambda}$.
Instead, it looks like $u_{\alpha_+,\beta_+,\lambda}$, a 1-soliton
whose center is:
$$
x=c_{\alpha_+,\beta_+,\lambda}(t)=\frac{1}{8} \log
\left(\frac{711433-365712 \sqrt{3}}{4434199}\right)-t.
$$
Since they are both traveling to the left at the same speed, their
difference is a constant
\begin{eqnarray*}
c_{\alpha_+,\beta_+,\lambda}-c_{\alpha_-,\beta_-,\lambda}&=&\frac{1}{8}
\log \left(\frac{39608599-22720000
  \sqrt{3}}{855800407}\right)\approx-1.01413,
\end{eqnarray*}
which means that the observer would find the localized disturbance
after the collision is about $1.014$ units farther to the left than
expected.  (Similarly, the phase shift for the stationary soliton
would be positive, which is why it has moved to the right after the interaction.)\end{example}

\begin{remark}\label{rem:phaseshift}
The phase shift experienced by each of the two
solitary waves experience have opposite signs.  Like the classical
commutative KdV solitons, if one is shifted forwards then the other
shifts backwards.  However, unlike the original KdV solitons studied
by Zabusky and Kruskal \cite{ZK}, the phase shift  is not completely
determined by the velocities, as the next example illustrates.
\end{remark}

\begin{example}\label{examp:phaseshiftrange} A one parameter family of interesting
examples is the case in which
$$
\alpha=1 \qquad
\beta=1\qquad
\lambda=1+3\qi 
$$
$$
\hatalpha=\qj\qquad
\hatbeta=-\qi+c\qk\qquad\hbox{and}\qquad
\hatlambda=1+4\qi.
$$
For any value of $c\in\R$ this represents a 2-soliton
solution with one traveling to the right at speed $26$ and another at
speed $47$.  However, the phase shift that they 
experience depends on $c$: $\ln(\gamma)$ is positive for
$|c|>1/\sqrt{11}$, zero if $|c|=1/\sqrt{11}$ and negative if
$|c|<1/\sqrt{11}$.   This dramatically demonstrates the fact that in
the non-commutative case the phase shift depends on the coefficients
$\alpha$, $\beta$, $\hatalpha$ and $\hatbeta$.  That does not happen
in the commutative case, as one can observe by noting that all dependence on these coefficients
cancel from the formula for $\gamma$ in
Proposition~\ref{prop:PhaseShift} if the parameters commute.
\end{example}

\section{Concluding Remarks}

Although it is surely no more than a coincidence that the quaternions and 
the existence of solitary waves
 were both famously discovered beside British canals in the 19th
century, this paper has found interesting results by studying
quaternion-valued solutions to the KdV-equation that can be produced
using the Chen determinant.  This is both a generalization of and a
special case of other published research, as will be explained further below.

When the functions in the KdV-Darboux kernel $\F$ are all either
complex-valued or real-valued, then the solution $u_{\F}$ satisfies
\eqref{eqn:stdKdV} and many of the new results in this paper reduce to
well-known results.  For instance, singularities of soliton and
periodic complex-valued solutions to KdV have been studied in
Reference~\cite{complex} much as we considered the singularities of the
quaternion-valued solutions in Theorem~\ref{thm:sing}.  However, the
generalization to the non-commutative case handled here is
non-trivial.  Without Corollary~\ref{cor:sing} it was not at all
obvious that the singularities that can be found in complex-valued KdV
solutions would necessarily fail to exist in their non-commutative quaternionic
counterparts.  Moreover, as noted in
Example~\ref{examp:phaseshiftrange}, the phase shift in the
quaternion-valued 2-soliton depends on the coefficients as well as on
the exponents, something that is not true in the commutative case.

On the other hand, there are also many published papers which address
non-commutative solutions to integrable PDEs in more general
settings.  The KdV equation is merely one equation in the KdV
hierarchy, which is a reduction of the KP hierarchy.  Quaternions can
be viewed as being a special four-dimensional subspace of larger
matrix groups, which are then special cases of abstract
non-commutative rings.  With all of that in mind, the methods utilized
herein can be seen as simply being a special case of the more general
approaches found in papers such as References \cite{NC3,EGR,NC1,NC2}.  However,
limiting ourselves to this manageable situation allows us to study
details that would be difficult to notice and demonstrate in those
more general settings.  For instance, we were able to show that it was
sufficient to consider exponential functions of the form $e^{\lambda x
  + \lambda^3 t}$ where $\lambda=\lambda_0+\lambda_1\qi$ is a complex
number with non-negative components $\lambda_i$ (cf.\ Lemma~\ref{lem:complexlambda}).  Doing so was
essential in being able to state Theorem~\ref{thm:sing} (the result
about singularities) in an easily understandable way.  And, since $\H$
is a four-dimensional vector space, we were able to graph the
corresponding solutions as a super-position of graphs of four
real-valued functions.  Moreover, it is interesting to know that
quaternion-valued solutions to KdV can be written in terms of the Chen
determinant, a result that presumably would not generalize to
solutions with values in arbitrary non-commutative rings.

There is one relatively recent paper by Huang \cite{QKdV} which, like this one,
specifically addresses quaternion-valued soliton solutions to KdV.  However, Huang's paper only considers a small subset of
the solution types that were addressed above.  In particular, it only
looks at solutions that would come from KdV-Darboux kernels made up of
functions of the form $\phi_{\alpha,\beta,\lambda}$ where
$\alpha,\beta\in\H$ and $\lambda\in\R$.  Thus, it does not include
breather, rational, or periodic solutions.  Finally, although
Reference \cite{QKdV}
 does discuss ``interactions'' of the solutions, that term has a
very different meaning in that paper.  Here,
Proposition~\ref{prop:LR} is viewed as a means to understand the
interaction of different solutions, with special emphasis on the
$2$-soliton solutions as representing the interaction between two
separate $1$-solitons (cf. Proposition~\ref{prop:PhaseShift}).  But in
Reference \cite{QKdV}, ``interaction'' refers to an algebraic structure that
Huang studies whereby two $n$-solitons can be combined to produce
another $n$-soliton (for the same fixed value of $n$).

There are many interesting examples which can be made using the
methods described above that we did not have the time or space to
present here. For instance, there are $2$-soliton solutions that
look like a combination of non-singular $1$-solitons at negative times
 and then like a pair of singular $1$-solitons for positive times (as
if the collision produced the singularities).   

There are also open problems that we have not been able to fully
address.  Theorem~\ref{thm:sing} completely determines when a solution
of the form
$u_{\alpha,\beta,\lambda}$ is singular.  However, it is not entirely
clear when combinations of such solutions are singular.  Although
Propositions~\ref{prop:LR} and \ref{prop:PhaseShift} may tell us when
they \textit{look} singular, as Remark~\ref{rem:non-sing-hard} shows,
that is not quite the same as actually being singular.  We do not yet
have any prediction for what $u_{\F}$ will look like if
$\F=\{\phi_{\alpha_1,\beta_1,\lambda_1},\ldots,\phi_{\alpha_n,\beta_n,\lambda_n}\}$
with $n>1$ and $\lambda_i=\lambda_{i1}\qi$ purely imaginary.  Most
intriguingly, since the solutions above were written in terms of the
Chen determinant of a Wronskian matrix, it would be interesting to
know whether there is a quaternionic analogue of the $\tau$-function
and Hirota's bilinear approach to soliton equations.

\section*{Acknowledgments} This paper grew out 
of a student research project at the College of Charleston
conducted in Summer 2018.  The first author's research was supported by a grant
from the School of Sciences and Mathematics, the third author's
research was supported by the University of Charleston, South Carolina
(the Graduate School at the College of Charleston), and the fourth
author's research was supported by a Summer Undergraduate Research
with Faculty (SURF) grant from the Office of Undergraduate Research
and Creative Activities.  Assistance from all of these offices and
from the Department of Mathematics at the College of Charleston is
greatly appreciated.

\end{document}